\documentclass[11pt,reqno]{amsart}
\usepackage{amsmath,amssymb,a4wide}

\usepackage{amsthm}
\usepackage{xspace}
\usepackage{graphicx}
\usepackage{mathrsfs}
\usepackage{xcolor}

\numberwithin{equation}{section}  

\newtheorem{punkt}{}[section]

\theoremstyle{plain}
\newtheorem{corollary}[punkt]{Corollary}
\newtheorem{lemma}[punkt]{Lemma}
\newtheorem{proposition}[punkt]{Proposition}
\newtheorem{theorem}[punkt]{Theorem}

\theoremstyle{definition}
\newtheorem{remark}[punkt]{Remark}

\newtheorem{example}[punkt]{Example}

\theoremstyle{plain}
\newtheorem*{corollary*}{Corollary}
\newtheorem*{lemma*}{Lemma}
\newtheorem*{proposition*}{Proposition}
\newtheorem*{theorem*}{Theorem}

\theoremstyle{definition}
\newtheorem*{remark*}{Remark}
\newtheorem*{example*}{Example}
\newtheorem*{definition*}{Definition}

\newcommand{\CC}{\mathbb{C}}
\newcommand{\RR}{\mathbb{R}}
\newcommand{\NN}{\mathbb{N}}

\newcommand{\BB}{\mathbb{B}}
\newcommand{\leb}{\lambda}  

\newcommand{\myo}{\mathcal{O}}
\newcommand{\mys}{\mathcal{S}}

\newcommand{\ee}{\mathbb{E}}
\newcommand{\pp}{\mathbb{P}}

\def\per{\qopname\relax{no}{per}}

\def\im{\qopname\relax{no}{Im}\,}

\def\eg{e.g.\@\xspace}
\def\ie{i.e.\@\xspace}
\def\iid{i.i.d.\@\xspace}

\newcommand{\ts}{\hspace{1pt}}

\def\i{{\operatorname{i}}}
\def\d{{\,\operatorname{d}}}

\def\var{\mathbb{V}\!\operatorname{ar}}

\def\eox{\hfill$\lozenge$}

\def\MM#1{\mu^{(#1)}}

\def\RMM#1{\mu_{\text{red}}^{(#1)}}
\def\RCM#1{\kappa_{\text{red}}^{(#1)}}
\def\FMM#1{\mu^{\bullet(#1)}}

\def\RFMM#1{\mu_{\text{red}}^{\bullet(#1)}}
\def\RFCM#1{\kappa_{\text{red}}^{\bullet(#1)}}
\def\CF#1{\varrho^{}_{#1}}

\begin{document}

\title{Diffraction Theory of Point Processes: \\[2mm] 
   Systems With Clumping and Repulsion}

\author{Michael Baake, Holger K\"osters}
\address{Fakult\"at f\"ur Mathematik, 
Universit\"at Bielefeld,
Postfach 100131, 33501 Bielefeld, Germany}
\email{$\{$mbaake,hkoesters$\}$@math.uni-bielefeld.de}

\author{Robert V. Moody}
\address{Department of Mathematics and Statistics, 
University of Victoria, 
Victoria, B.C.,
Canada V8W 2Y2}
\email{rmoody@uvic.ca}

\date{November 24th, 2014}

\begin{abstract}
We discuss several examples of point processes
(all taken from~\cite{HKPV:2009}) 
for which the autocorrelation and diffraction measures
can be calculated explicitly.
These include certain classes
of determinantal and permanental point processes,
as well as an isometry-invariant point process that arises 
as the zero set of a Gaussian random analytic function.
\end{abstract}

\maketitle

\markboth{Diffraction Theory of Point Processes}{Diffraction Theory of
  Point Processes}

\medskip

\section{Introduction}
\label{sec:introduction}

Mathematical diffraction theory deals with the relationship
between the structure of point configurations in space 
and the associated autocorrelation and diffraction measures,
one of the main questions being how the order properties 
of the point configuration translate into properties
of the diffraction measure.
There exist many results about deterministic point configurations 
(periodic and aperiodic tilings, model sets, substitution systems,
compare \cite{baake-grimm:2013} and the references therein),
but in recent years random point configurations have also been considered.
In particular, reference \cite{baake-birkner-moody:2010} provides a
general framework for the investigation of point configurations within
the theory of point processes, along with a number of examples, most
of them closely connected to renewal and Poisson processes.  
However, beyond the i.i.d.\@ situation,
the number of explicit examples is still rather small, 
and it seems that more examples are needed for a better understanding 
of the problem and a further development of the theory. 
In particular, it seems desirable to have examples with some effective interaction
for which the autocorrelation and diffraction measures may be calculated explicitly.
Some of the few exceptional cases which have been considered here
are the Ising model on the square lattice 
and the dimer model on the triangular lattice; 
see \cite{baake-hoeffe:2000}.

The aim of this paper is to extend the list of these explicit examples
by discussing various point processes (all taken from~\cite{HKPV:2009}) 
from the viewpoint of mathematical diffraction theory.
All these examples are simple, stationary and ergodic (see
Section~2 for definitions).  Moreover, the numbers of points in
neighbouring subsets of Euclidean space may be either positively
correlated (``clumping'') or negatively correlated (``repulsion'').
We discuss determinantal and permanental point processes
as (classes of) examples for systems with repulsion and clumping, respectively.
As a further example for a system with repulsion, we consider the zero set 
of a certain Gaussian random analytic function in the complex plane. 
More precisely, we take the unique Gaussian random analytic function, up~to scaling, 
such that the zero set is translation-invariant (in~fact, even isometry-invariant)
in distribution; see Section~6 for details.
Furthermore, we also briefly look at Cox processes.

Our main results support the widespread expectations about
the diffraction measures of generic random point configurations 
(with good mixing properties, say).
For instance, in~most of our examples, the diffraction measure
is absolutely continuous apart from the (trivial) Bragg peak at the origin.
(Here, by a \emph{Bragg peak}, one understands a point mass 
contained in the diffraction measure.)
Moreover, for a certain class of determinantal and permanental point processes, it turns out 
that the ``diffraction spectrum'' is equivalent to the ``dynamical spectrum''; 
see Remark~\ref{rem:dynamical-spectrum} for details.
For general point processes, the question under what conditions 
this is true seems to be open,
but for the above examples, this is perhaps not too surprising
in view of their excellent mixing properties.

\section{Preliminaries}

This section contains some background information on point processes,
Fourier transforms, and mathematical diffraction theory.  As they are
sufficient for our~discussion here, we restrict ourselves to positive
measures, and refer to \cite{baake-birkner-moody:2010} and the
references therein for the general case.

\subsection{Point processes}

A measure on $\RR^d$ is a measure on the Borel $\sigma$-field
(or $\sigma$-algebra) $\BB^d$ of~$\RR^d$.  A measure $\omega$ on $\RR^d$ is
\emph{locally finite} if $\omega(A) < \infty$ for any bounded Borel
set $A$.  A \emph{point measure} on $\RR^d$ is a measure on
$\RR^d$ taking values in $\NN_0 \cup \{ \infty \}$.  A point
measure $\omega$ on $\RR^d$ is \emph{simple} if $\omega(\{x\})
\leq 1$ for any $x \in \RR^d$.  We write $\mathcal{M}(\RR^d)$
for the space of locally finite measures on $\RR^d$ and
$\mathcal{N}(\RR^d)$ for the subspace of locally finite point
measures on $\RR^d$.  It is well known (compare \cite[Appendix
A2]{DV1} and \cite[Section 9]{DV2}) that there exists a metric such
that $\mathcal{M}(\RR^d)$ and $\mathcal{N}(\RR^d)$ are
complete separable metric spaces, the induced topology is that of
vague convergence, and the induced Borel $\sigma$-fields
$\mathscr{M}(\RR^d)$ and $\mathscr{N}(\RR^d)$ are the smallest
$\sigma$-fields such that the mappings $\omega \mapsto \omega(A)$, 
with $A \in \BB^d$, are measurable.  A~\emph{random measure} 
is a random variable taking values in $\mathcal{M}(\RR^d)$, 
and a~\emph{point process} is a random variable taking values 
in $\mathcal{N}(\RR^d)$.  A point process is called \emph{simple} 
if its realisation is simple with probability $1$.  
If $\omega$ is a random measure or a point process such that 
$\ee(\omega(A)) < \infty$ for any bounded Borel set $A$, 
we say that the expectation measure of $\omega$ exists, 
and call the~measure $A \mapsto \ee(\omega(A))$ 
the \emph{expectation measure} of $\omega$.

A locally finite point measure $\omega$ on $\RR^d$ may be written
in the form $\omega = \sum_{i \in I} \delta_{x_i}$, where $I$ is~a
countable index set and $(x_i)_{i \in I}$ is a~family of points in
$\RR^d$ with at most finitely points in any bounded Borel set.
Then, for any $k \geq 1$, the locally finite point measures $\omega^k$
and $\omega^{\bullet k}$ on $(\RR^d)^k$ are defined by
\[
  \omega^k \, := 
  \sum_{i_1,\hdots,i_k \in I} \delta_{(x_{i_1},\hdots,x_{i_k})}
  \qquad\text{and}\qquad
  \omega^{\bullet k} \, := 
  \sum_{\substack{i_1,\hdots,i_k \in I \\ 
  \text{distinct}}} \delta_{(x_{i_1},\hdots,x_{i_k})} \,.
\]
Note that $\omega^k$ is simply the $k$-fold product measure of
$\omega$.  If $\omega$ is a point process such that $\ee(\omega(A)^k)
< \infty$ for any bounded Borel set $A$, the expectation measures
$\MM{k}$ and $\FMM{k}$ of $\omega^k$ and $\omega^{\bullet k}$ are
called the \emph{$k$th moment measure} of $\omega$ and the \emph{$k$th
  factorial moment measure} of $\omega$, respectively.  If $\FMM{k}$
is absolutely continuous with~respect~to Lebesgue measure
on~$(\RR^d)^k$, its density is denoted by $\CF{k}$ and called
the \emph{$k$-point correlation function} of the point process
$\omega$.

\begin{remark}
  It is easy to see that the absolute continuity of the second
  factorial moment measure implies that the underlying point process
  is simple; compare \cite[Proposition~5.4.6]{DV1} for details.  \eox
\end{remark}

For any $x \in \RR^d$, we write $T_x$ for the translation by $x$ 
on~$\RR^d$ (with $T_x(u) := u + x$).
The notation $T_x$ is extended to the induced translation actions
on sets, functions, and measures in the standard way:
$T_x(A) := x+A := \{ T_x(y) : y \in A \}$, 
$(T_x f)(y) := f(T_{-x}(y)) = f(-x+y)$,
and $(T_x \omega)(A) := \omega (T_{-x} (A)) = \omega(-x+A)$.
Note that the latter is consistent with translation of an indicator function $\pmb{1}_A$
(also interpretable as a measure) so that $T_x(\pmb{1}_A) = \pmb{1}_{T_x(A)}$.
Finally, we also write $T_x$ for the induced translation on sets of measures: 
$T_x(B) := \{ T_x(\omega) : \omega \in B \}$.

A set $B \in \mathscr{M}(\RR^d)$ is called \emph{invariant} if
$T_x^{-1}(B) = B$ for any $x \in \RR^d$.  A~random measure
$\omega$ on $\RR^d$ is called \emph{stationary} (or
\emph{translation-invariant}) if for any $x \in \RR^d$,
$\omega$~and~$T_x(\omega)$ have the same distribution.  A~random
measure $\omega$ on $\RR^d$ is called \emph{ergodic} if it is
stationary and if, for any invariant set $B \in
\mathscr{M}(\RR^d)$, one has $\pp(\omega \in B) \linebreak[2] \in
\{ 0,1 \}$.  A~random measure $\omega$ on $\RR^d$ is called
\emph{mixing} if it is stationary and if, for any sets $B_1,B_2 \in
\mathscr{M}(\RR^d)$, one has $\pp(\omega \in T_x^{-1}(B_1) \cap
B_2) \longrightarrow \pp(\omega \in B_1) \, \pp(\omega \in B_2)$ as
$|x| \to \infty$.  Here, $|x|$ denotes the Euclidean norm of $x$.  
It is well known that mixing implies ergodicity; see \eg \cite[Section 12.3]{DV2}.

If $\omega$ is a stationary point process such that $\ee(\omega(A)^k)
< \infty$ for any bounded Borel set $A$, the \emph{reduced $k$th
  moment measure} $\RMM{k}$ of $\omega$ and the \emph{reduced $k$th
  factorial moment measure} $\RFMM{k}$ of $\omega$ are the (unique)
locally finite measures on $(\RR^{d})^{k-1}$ such that
\[
  \int_{(\RR^{d})^{k}} f(x_1,x_2,\hdots,x_k) \d{\MM{k}}(x) \, = 
  \int_{\RR^{d}} \int_{(\RR^{d})^{k-1}} f(x,x+y_1,\hdots,x+y_{k-1}) \d{\RMM{k}}(y) \d{x}
\]
and
\[
  \int_{(\RR^{d})^{k}} f(x_1,x_2,\hdots,x_k) \d{\FMM{k}}(x) \, = 
  \int_{\RR^{d}} \int_{(\RR^{d})^{k-1}} f(x,x+y_1,\hdots,x+y_{k-1}) \d{\RFMM{k}}(y) \d{x}
\]
for any bounded measurable function $f$ on $(\RR^d)^k$ with
bounded support, compare  \cite[Proposition~12.6.3]{DV2}.
Here, the first reduced (factorial) moment measure is regarded as a
constant $\varrho$, which is also called the \emph{mean density} of
the point process.  Furthermore, if $\FMM{k}$ is absolutely
continuous with respect to Lebesgue measure on~$(\RR^d)^k$, we may
assume its density $\CF{k}(x_1,\hdots,x_k)$ to be
\emph{translation-invariant}, i.e.
\begin{equation}\label{eq:translation-invariance}
  \CF{k}(x_1+t,\hdots,x_k+t) \, = \, \CF{k}(x_1,\hdots,x_k)
\end{equation}
for any $t \in \RR^d$.  Then, $\RFMM{k}$ is absolutely continuous
with respect to Lebesgue measure on~$(\RR^d)^{k-1}$, with density
$ \CF{k}(0,y_1,\hdots,y_{k-1}) \,.  $ In particular, $\CF{1}(0) =
\varrho$.

\subsection{Fourier transforms}
For the Fourier transform of a function $f \in L^1(\RR^d)$, we use
the convention
\begin{equation}\label{eq:fouriertransform}
  \widehat{f}(y) \, = \int_{\RR^d} f(x) e^{-2\pi\i xy} \d{x} \ts ,
\end{equation}
where $xy$ is the standard inner product on $\RR^d$.  Let us note
that, with this con\-vention, Fourier inversion takes the form
\begin{equation}\label{eq:fourierinversion}
  f(x) \, = \int_{\RR^d} \widehat{f}(y) e^{2\pi\i xy} \d{y}
\end{equation}
for any continuous function $f \in L^1(\RR^d)$ such that
$\widehat{f} \in L^1(\RR^d)$.
The special case of radially symmetric functions will be discussed later
when we need it.

Besides Fourier transforms of $L^1$-functions, we will also use
Fourier transforms of $L^2$-functions.  The Fourier transform on $L^2$
is defined as the (unique) continuous extension of the Fourier
transform on $L^1$ restricted to $L^1 \cap L^2$, viewed as a mapping
from $L^1 \cap L^2 \subset L^2$ to $L^2$.  It is well known that the
Fourier transform on $L^2$ is an~isometry.

Moreover, we will also use Fourier transforms of translation-bounded 
measures.  A measure $\mu$ on $\RR^d$ is
\emph{translation-bounded} if, for any bounded Borel set $B \in \BB^d$, 
$\sup_{x \in \RR^d} \mu(x + B) < \infty$.  
A translation-bounded measure $\mu$ on $\RR^d$ is
\emph{transformable} if there exists a translation-bounded measure
$\widehat\mu$ on $\RR^d$ such that
\[
  \int_{\RR^{d}} f(x) \d\widehat\mu(x) \, = 
  \int_{\RR^{d}} \widehat{f}(x) \d\mu(x)  
\]
holds for any Schwartz function $f \in \mys(\RR^d)$.  
In this case, the measure $\widehat\mu$ is unique, and it~is~called 
the \emph{Fourier transform} of the measure $\mu$.  Indeed, 
these definitions may even be extended to \emph{signed} measures; 
see \cite{baake-birkner-moody:2010} for details.

Given a locally integrable function $f$ on $\RR^d$, we write $f \,
\leb^d$ or $f(x) \, \leb^d$ for the signed measure on
$\RR^d$ given by $(f \, \leb^d)(B) := \int_B f(x) \d{x}$,
with $B \in \BB^d$ bounded.  Note that if $f$ is integrable, 
$(f \, \leb^d)\widehat\enspace = \widehat{f} \, \leb^d$.
Note further that, if $f$ is the Fourier transform of some integrable
function $\varphi$, it follows by standard Fourier inversion that 
$(f \, \leb^d)\widehat\enspace = \varphi^{}_{-} \, \leb^d$ where
$\varphi^{}_{-}(x) := \varphi(-x)$ denotes the reflection of $\varphi$ 
at the origin.

\subsection{Mathematical diffraction theory}

Let $\omega$ be a locally finite measure on~$\RR^d$, and let
$\widetilde\omega(A) := \omega(-A)$ be its reflection at the origin.
Write $B_n$ for the open ball of radius $n$ around the origin and
$\leb^d(B_n)$ for its $d$-dimensional volume.  The
\emph{auto\-correlation measure} of $\omega$ is defined by
\[
  \gamma \, := \lim_{n \to \infty} 
  \frac{\omega|_{B_n} \ast \widetilde{\omega|_{B_n}}}{\leb^d(B_n)} \ts ,
\]
provided that the limit exists in $\mathcal{M}(\RR^d)$
with respect to the vague topology.
In this case, the \emph{diffraction measure} of $\omega$ is the Fourier
transform of $\gamma$.  Let us note that $\widehat\gamma$ exists due
to the fact that $\gamma$ is a positive and positive-definite measure,
and that $\widehat\gamma$ is also a positive and positive-definite
measure; see \cite{baake-birkner-moody:2010}, \cite{baake-grimm:2013},
\cite{berg-forst:1975} or \cite[Section 8.6]{DV1} for details.

In each of the following examples, $\omega$ will be given by the realisation
of a stationary and ergodic simple point process on $\RR^d$.  The
following result from \cite{baake-birkner-moody:2010}, see also \cite{lenz-strungaru:2009},
shows that the auto\-correlation and diffraction measures exist and, moreover,
almost surely do not depend on the~realisation.

\begin{theorem}[{\cite[Theorem 1.1]{gouere:2003}, 
\cite[Theorem 3]{baake-birkner-moody:2010}}]\label{thm:ET}
Let $\omega$ be a stationary and ergodic point process such that the
reduced first~moment measure $\RMM{1} = \varrho$ and the reduced
second~moment measure $\RMM{2}$ exist.  Then, almost surely, the
auto\-correlation measure $\gamma$ of\/ $\omega$ exists and satisfies
\[
  \gamma \, = \, \RMM{2} \, = \, \varrho \delta_0 + \RFMM{2} \ts .
\]
\qed
\end{theorem}

\begin{remark}
  In \cite[Theorem 1.1]{gouere:2003} and \cite[Theorem
  3]{baake-birkner-moody:2010}, the preceding result is stated in a
  slightly different form, namely that the autocorrelation measure of
  $\omega$ is equal to the first moment measure of the so-called Palm
  measure of $\omega$.  However, 
  the latter coincides with the reduced second moment measure of~$\omega$ 
  under the assumptions of the theorem; 
  see \eg \linebreak[1] \cite[Equation (47)]{baake-birkner-moody:2010}. 
  The second equality in Theorem \ref{thm:ET} is a well-known relation
  between $\RMM{2}$ and $\RFMM{2}$; see e.g.\@ \cite[Section 8.1]{DV1}.
  \eox
\end{remark}

We will often use Theorem~\ref{thm:ET} in the following form.

\begin{corollary}\label{cor:ET}
  Suppose that, in addition to the assumptions of Theorem
  $\ref{thm:ET}$, the factorial moment measure $\FMM{2}$ is absolutely
  continuous with a translation-invariant density of the form
\begin{equation}\label{eq:extra-2}
  \CF{2}(x_1,x_2)\, = \, \varrho^2 + g(x_2-x_1) 
  \qquad \text{with $g \in L^1(\RR^d)$.}
\end{equation}
Then, almost surely, the autocorrelation and diffraction measures of\/
$\omega$ exist and are given by
\begin{equation}\label{eq:auto-3}
  \gamma \, = \, \varrho\ts \delta_0 + (\varrho^2 + g) \ts \leb^d
\end{equation}
and
\begin{equation}\label{eq:diff-3}
  \widehat\gamma \, = \, \varrho^2 \delta_0 + 
    \bigl(\varrho + \widehat{g}\,\bigr) \ts \leb^d \ts ,
\end{equation}
respectively.
\end{corollary}

If the argument of the function needs to be specified, 
we usually write $g(x) \ts \leb^d$ in~\eqref{eq:auto-3}
and $\widehat{g}(t) \ts \leb^d$ in \eqref{eq:diff-3}.
Note that, under the assumptions of Corollary~\ref{cor:ET}, 
the diffraction measure is absolutely continuous apart from 
the Bragg peak at the origin.

\begin{proof}[Proof of Corollary $\ref{cor:ET}$]
  Eq.~\eqref{eq:auto-3} is immediate from Theorem~\ref{thm:ET} and our
  comments around Eq.~\eqref{eq:translation-invariance}.
  Eq.~\eqref{eq:diff-3} then follows by taking the Fourier transform
  and using the relations $\widehat{\delta_0} = \leb^d$,
  $\widehat{\leb^d} = \delta_0$, 
  and $\widehat{g \, \leb^d} = \widehat{g} \, \leb^d$.
\end{proof}

\section{Determinantal point processes}
\label{sec:DPP}

Determinantal point processes are used to model particle
configurations with repulsion; see \cite{soshnikov:2000},
\cite[Section 4.2]{HKPV:2009} or \cite[Section 4.2]{AGZ:2010} for
background information.  See also \cite{macchi:1975} and \cite[Example
5.4\,(c)]{DV1}, where these point processes are called \emph{fermion
  processes}.

In the sequel, we shall always assume the following:
\begin{align}
\label{eq:main}
\begin{minipage}{0.8\textwidth}
  The kernel $K : \RR^d \times \RR^d \longrightarrow \CC$ 
  is continuous, Hermitian and positive-definite.
\end{minipage}
\end{align}
Here, $K$ is \emph{Hermitian} if, for all $x,y \in \RR^d$, $K(x,y) =
\overline{K(y,x)}$, and \emph{positive-definite} if, for all $k \in
\NN$ and all $x_1,\hdots,x_k \in \RR^d$, $\det \bigl(K(x_i,x_j)^{}_{1 \leq
  i,j \leq k} \bigr) \geq 0$.

A point process $\omega$ on $\RR^d$ is called \emph{determinantal}
with kernel $K$ if, for any $k \in \NN$, the $k$-point correlation
function exists and is given by
\begin{equation}\label{eq:DPP-def}
  \CF{k}(x_1,\hdots,x_k) \, = \,
   \det \bigl( K(x_i,x_j)^{}_{1 \leq i,j \leq k} \bigr) .
\end{equation}

It is well known that, if there exists a determinantal point process
with a given kernel $K$, its distribution is uniquely determined;
compare \cite[Lemma 4.2.6]{HKPV:2009}.  As~regards existence, note
first that, if $K$ is a kernel as in Eq.~\eqref{eq:main}, then, for any
compact subset $B \subset \RR^d$, we have an~integral operator
$\mathcal{K}_{B} : L^2(B) \to L^2(B)$ defined by
\begin{equation}\label{eq:operator}
  \bigl(\mathcal{K}_{B} f\bigr)(x) \, := 
  \int_B K(x,y) f(y) \d{y} \qquad (x \in B) \ts .
\end{equation}
It is well known that this operator is bounded, self-adjoint,
positive-definite and of trace class; see \cite[Lemma
4.2.13]{AGZ:2010} and references given there.  Furthermore, there is
the following criterion for the existence of an associated
determinantal point process.

\begin{theorem}[{\cite[Theorem 3]{soshnikov:2000}, 
\cite[Theorem~4.5.5]{HKPV:2009}, \cite[Corollary~4.2.22]{AGZ:2010}}]
\label{thm:DPP-existence}
Let $K$ be a kernel as in Eq.~\eqref{eq:main}.  Then $K$ defines a
determinantal point process on $\RR^d$ if and only if, for any compact
subset $B \subset \RR^d$, the spectrum of the operator
$\mathcal{K}_{B}$ is~contained in the interval\/ $[0,1]$.  \qed
\end{theorem}

\begin{remark}
  Let us mention that in part of the above-mentioned literature it is
  assumed that the kernel $K$ is measurable, locally
  square-integrable, Hermitian, positive-definite, and locally of
  trace class.  Indeed, all the results stated above continue to hold
  under this weaker assumption.  However, we will only be interested
  in stationary determinantal point processes, and the assumption of
  continuity is satisfied in all our examples.  Furthermore,
  continuous kernels are convenient in~that they give rise to (unique)
  continuous correlation functions.  \eox
\end{remark}

Henceforward, we shall always assume that the determinantal point
process $\omega$ is \emph{stationary} with mean density $1$.  Then, by
Eq.~\eqref{eq:translation-invariance}, (the continuous versions of)
the first and second correlation functions satisfy
\begin{equation}\label{eq:DPP-first-and-second}
  \CF{1}(x) \, = \, K(x,x) \, = \, 1 
  \quad\text{and}\quad
  \CF{2}(x,y) \, = \, 1 - |K(x,y)|^2 \, = \, 1 - g(x-y)
\end{equation}
for all $x,y \in \RR^d$, where $g(x) := |K(0,x)|^2$.  
Note that $g$ is positive and positive-definite;
for the latter, use that $g(x-y) = |K(x,y)|^2 = K(x,y) \overline{K(x,y)}$
and that the pointwise (or Hadamard) product
of positive-definite kernels is also positive-definite.

\begin{remark}
  Let us emphasise that the stationarity of the determinantal point
  process entails the translation-invariance of the correlation
  functions and of the modulus of the kernel, but not necessarily 
  that of the kernel itself, as for the Ginibre process (see
  Example~\ref{ex:ginibre}).  \eox
\end{remark}

\begin{lemma}\label{lemma:integrability}
  Let\/ $\omega$ be a stationary determinantal point process with a
  kernel\/ $K$ as specified in Eq.~\eqref{eq:main} and with mean
  density\/ $1$, and let\/ $g$ be as in
  Eq.~\eqref{eq:DPP-first-and-second}.  Then, $g$ is integrable with\/
  $\int_{\RR^d} g(y) \d{y} \leq 1$.
\end{lemma}

\begin{proof}
  We use the same argument as in the proof of \cite[Lemma
  4.2.32]{AGZ:2010}.  It follows from the definitions of the ordinary
  and the factorial moment measures and Eq.~\eqref{eq:DPP-first-and-second} 
  that, for any bounded Borel set $A$, one has
\begin{align*}
   0 \, \leq \, \var(\omega(A)) 
   \, & = \, \MM{2}(A \times A) - \big(\MM{1}(A)\big)^2
   \, = \, \MM{1}(A) + \FMM{2}(A \times A) - \big(\MM{1}(A)\big)^2 \quad \\[1mm]
   & = \int_A 1 \d{x} + \int_A \int_A \big(\varrho_2(x,y) - 1\big) \d{y} \d{x} 
   \, = \, \leb^d(A) - \int_A \int_{A+x} g(y) \d{y} \d{x} \ts .
\end{align*}

\noindent{}Taking $A = B_n$, the ball of radius $n$ around the origin,
we get
\[
   \frac{1}{\leb^d(B_{n})} \int_{B_{n}} 
   \int_{B_{n}+x} g(y) \d{y} \d{x} \, \leq \, 1
\]
for any $n \in \NN$, from which it follows that $\int_{\RR^d} g(y)
\d{y} \leq 1 $.
\end{proof}

Combining Lemma~\ref{lemma:integrability} with Corollary~\ref{cor:ET},
we get the following result.

\begin{proposition}\label{prop:determinantal}
  Let\/ $\omega$ be a stationary and ergodic determinantal point process
  with a kernel\/ $K$ as in Eq.~\eqref{eq:main} and with mean density\/
  $1$.  Then, the autocorrelation and diffraction measures of\/
  $\omega$ are given by
\[
  \gamma \, = \, \delta_0 + ( 1 - g ) \, \leb^d
  \qquad\text{and}\qquad
  \widehat\gamma \, = \, \delta_0 + ( 1-\widehat{g}\, ) \, \leb^d \ts ,
\]
  with\/ $g$ as in Eq.~\eqref{eq:DPP-first-and-second}. Moreover, 
  we have\/ $0 \leq \widehat{g}(t) \leq 1$ for all\/ $t \in \RR^d$,
  with $\widehat{g}(t) = 1$ at most for\/ $t = 0$.
  In particular, the absolutely continuous part of the diffraction measure 
  is equivalent to Lebesgue measure.
\end{proposition}

\begin{proof}
  The statements about $\gamma$ and $\widehat\gamma$ are immediate
  from Corollary~\ref{cor:ET}, Eq.~\eqref{eq:DPP-first-and-second},
  and Lemma~\ref{lemma:integrability}.  The statements about $\widehat{g}$ 
  follow from the positivity and positive-definiteness~of the function $g$
  and well-known properties of the~Fourier transform.
\end{proof}

Let us now turn to the case that the kernel $K$ itself is
\emph{translation-invariant}, which means that there exists a function
$K \! : \, \RR^d \to \CC$ such that $K(x,y) = K(x-y)$ for all $x,y
\in \RR^d$.  (By slight abuse of notation, we use the same symbol for
the function and for the~associated kernel.)  More precisely, we will
assume the following:
\begin{align}\label{eq:DPP-ti}
  \begin{minipage}{0.8\textwidth}
    We have $K(x,y) = K(x-y)$ for all $x,y \in \RR^d$, where 
    the function \linebreak $K \! : \, \RR^d \to \CC$ 
    on the right-hand side is the Fourier transform 
    of a~probability density $\varphi$ on $\RR^d$ 
    with values in $[0,1]$.
  \end{minipage}
\end{align}
Note that Condition \eqref{eq:DPP-ti} entails Condition
\eqref{eq:main}.

\begin{remark}
  It can be shown that, if a kernel $K$ as in Eq.~\eqref{eq:main} is
  translation-invariant, it~is the kernel of a (stationary)
  determinantal point process with mean density $1$ if and only if 
  it~is of the form in Eq.~\eqref{eq:DPP-ti};
  cf. \cite{lavancier-moller-rubak:2012} for a similar result.  \eox
\end{remark}

Let us sketch the argument why a kernel $K$ as in
Eq.~\eqref{eq:DPP-ti} defines a stationary and ergodic determinantal
point process with mean density $1$.  Suppose that
Eq.~\eqref{eq:DPP-ti} holds.  Then, $\varphi \in L^1(\RR^d) \cap
L^2(\RR^d)$, $\widehat\varphi = K \in L^2(\RR^d) \cap
L^\infty(\RR^d)$, and $\widehat{K} = \varphi^{}_{-}$ by Fourier
inversion (in the $L^2$-sense),
where $\varphi^{}_{-}(x) = \varphi(-x)$ as before.
Moreover, since $K \in L^2(\RR^d)$, the convolution
\[
  \bigl(\mathcal{K}f\bigr)(x)\, := \,
  \bigl(K \ast f \bigr)(x) \, := \int_{\RR^{d}} K(x-y) f(y) 
  \d{y} \qquad (x \in \RR^d)
\]
is well-defined for any $f \in L^2(\RR^d)$ by the Cauchy--Schwarz
inequality.  Furthermore, for $f \in L^1(\RR^d) \cap L^2(\RR^d)$, we
have
\[
  \widehat{\mathcal{K} f}\, = \, 
  \widehat{K \ast f} \, = \, 
  \widehat{K} \cdot \widehat{f} \, = \,
  \varphi^{}_{-} \cdot \widehat{f}
\]
and therefore
\begin{equation}\label{eq:multiplication}
   \mathcal{K} f \, = \, 
   \mathcal{F}^{-1} M_{\varphi^{}_{-}} \mathcal{F} f \ts ,
\end{equation}
where $\mathcal{F}$ denotes the Fourier transform on $L^2(\RR^d)$,
$\mathcal{F}^{-1}$ its inverse, and $M_{\varphi^{}_{-}}$
the~multiplication operator $g \mapsto g \varphi^{}_{-}$ on
$L^2(\RR^d)$.  By continuity, this~extends to all of $L^2(\RR^d)$, and
$\mathcal{K}$ is a bounded, self-adjoint, positive-definite
convolution operator on $L^2(\RR^d)$ with spectrum $S_\varphi \subset
[0,1]$, where
\[
  S_\varphi \, := \,
  \{ y \in [0,1] : \leb^d(\{ x \in \RR^d : 
  |\varphi(x)-y| < \varepsilon  \}) > 0 \ 
  \text{for all $\varepsilon > 0$} \} 
\]
is the essential range of $\varphi$.  Thus, the operators
$\mathcal{K}_{B}$ defined in Eq.~\eqref{eq:operator} must also have
spectra contained in $[0,1]$, and it follows from Theorem
\ref{thm:DPP-existence} that $K$ defines a~determinantal point process
$\omega$.

Since the correlation functions determine the distribution of
$\omega$, $\omega$ is clearly stationary with mean density $K(0) = 1$.
Moreover, since $K = \widehat\varphi$ with $\varphi \in L^1(\RR^d)$,
it~follows from the Riemann--Lebesgue lemma that $K(x) \rightarrow 0$
as $|x| \to \infty$, and this implies that $\omega$ is mixing and
ergodic; see \cite[Theorem 7]{soshnikov:2000} or \cite[Theorem
4.2.34]{AGZ:2010}.

We therefore obtain the following consequence of
Proposition~\ref{prop:determinantal}.

\begin{corollary}\label{cor:determinantal}
  If\/ $K$ is a kernel as in Eq.~\eqref{eq:DPP-ti}, the
  autocorrelation and diffraction measures of the associated
  determinantal point process are given by
\[
   \gamma \, = \, \delta_0 + (1 - |K|^2) \, \leb^d
   \qquad\text{and}\qquad
   \widehat\gamma \, = \, 
   \delta_0 + (1-\widehat{|K|^2}) \, \leb^d \ts .
\]
Equivalently,
\[
  \gamma \, = \, \delta_0 + (1 - |\widehat\varphi|^2) \, \leb^d
  \qquad\text{and}\qquad
  \widehat\gamma \, = \,
   \delta_0 + (1-(\varphi \ast \varphi^{}_{-})) \, \leb^d \ts ,
\]
where $\varphi^{}_{-}(x) := \varphi(-x)$ as above.  \qed
\end{corollary}

\begin{remark}[Self-reproducing kernels]\label{rem:self-reproducing}
  Suppose that, in the situation of
  Proposition~\ref{prop:determinantal}, the kernel $K$ is
  \emph{self-reproducing} in the sense that
\begin{equation}\label{eq:self-reproduction}
   \int_{\RR^d} K(x,y) K(y,z) \d{y} \, = \, K(x,z)
\end{equation}
for all $x,z \in \RR^d$ or, equivalently, the associated integral
operator $\mathcal{K}$ on $L^2(\RR^d)$ (defined similarly as in
\eqref{eq:operator}, but with $B = \RR^d$) is a~projection, \ie
$\mathcal{K}^2 = \mathcal{K}$.  (Let~us~mention without proof that
$\mathcal{K}$ is indeed a well-defined operator on $L^2(\RR^d)$, as
follows from Theorem~\ref{thm:DPP-existence} and Lemma
\ref{lemma:integrability}.)  Then, with $g$ as in
Eq.~\eqref{eq:DPP-first-and-second}, we~have
\[
  \widehat{g}(0) \, = \int |K(0,x)|^2 \d{x} 
   \, = \int K(0,x) K(x,0) \d{x}
   \, = \, K(0,0) \, = \, 1 \ts ,
\]
so that the density of the absolutely continuous part of the
diffraction measure equals zero at the origin.

Moreover, for a translation-invariant kernel as in
Eq.~\eqref{eq:DPP-ti}, the converse is also true.  Indeed, in this
case, we have $\widehat{|K|^2}(0) = 1$ if and only if $\widehat{K}(t)$
is an~indicator function, as already pointed out in
\cite{soshnikov:2000}.  For the convenience of the reader, let us
reproduce the argument here: Using that $\widehat{f_1 f_2} = \widehat{f_1}
\ast \widehat{f_2}$ for $f_1,f_2 \in L^2(\RR^d)$ and that $\widehat{K} =
\varphi^{}_{-}$ is $[0,1]$-valued, we~obtain
\[
  \widehat{|K|^2}(0)
  \, =  \, \bigl(\widehat{K} \ast \widehat{\overline{K}}\,\bigr)(0)
  \, = \int \widehat{K}(t) \overline{\widehat{K}(t)} \, \d{t}
  \, = \int (\widehat{K}(t))^2 \, \d{t} 
  \, \leq \int \widehat{K}(t) \, \d{t}
  \, = \,  K(0) \, = \, 1 \ts ,
\]
with equality if and only if $\widehat{K} = 1$ holds a.e.\ on the set
$\{ \widehat{K} \ne 0 \}$.  Since indicator functions correspond to
projection operators by Eq.~\eqref{eq:multiplication}, this proves the
claim.  \eox
\end{remark}

\smallskip

By \cite[Corollary 4.2.23]{AGZ:2010},
`thinnings' of determinantal point processes
are again determinantal point processes.
This leads to the following observation.

\begin{remark}[Thinned determinantal point processes]\label{rem:thinning}
  Let $\omega$ be a determinantal point process on $\RR^d$ with a
  kernel $K$ as in Eq.~\eqref{eq:main}, let $0 < p \leq 1$, and
  let~$\omega_p$ denote the point process obtained from $\omega$ by
  (i) deleting each point with probability $1-p$, independently of one
  another, and (ii) rescaling the resulting point process so~that the
  mean density becomes $1$.  Then, $\omega_p$ is the determinantal
  point process associated with the kernel $K_p(x,y) :=
  K(x/p^{1/d},y/p^{1/d})$, as follows from \cite[Corollary
  4.2.23]{AGZ:2010}.  Thus, each determinantal point process $\omega$
  gives rise to an entire family $(\omega_p)_{0 < p \leq 1}$ of
  determinantal point processes.

  Furthermore, if $\omega$ is stationary and ergodic, $\omega_p$ is
  also stationary and ergodic, and if $g$ is defined as in
  Eq.~\eqref{eq:DPP-first-and-second}, and $g_p$ is the analogous
  function for $\omega_p$, we~have $g_p(x) = g(x/p^{1/d})$ and
  $\widehat{g}_p(t) = p \widehat{g}(tp^{1/d})$.  Therefore, by
  Proposition~\ref{prop:determinantal}, the autocorrelation and
  diffraction measures of $\omega_p$ are given by
\[
  \gamma_p \, = \, \delta_0 + \bigl(1 - g(x/p^{1/d})\bigr) \ts \leb^d
  \quad \text{and}\quad
  \widehat\gamma_p \, = \, 
  \delta_0 + \bigl(1-p \, \widehat{g}(tp^{1/d})\bigr) \ts \leb^d \ts .
\]
As $p \to 0$, the repulsion between the points decreases,
and the point process converges in distribution
to the homogeneous Poisson process with intensity $1$.

Finally, note that if $K$ is a translation-invariant kernel as in
Eq.~\eqref{eq:DPP-ti}, the same holds for $K_p$.  More precisely, if
$K$ is the Fourier transform of the probability density $\varphi(t)$
(with values in $[0,1]$), then $K_p$ is the Fourier transform of the
probability density $\varphi_p(t) := p \, \varphi(tp^{1/d})$ (with
values in $[0,p]$).  In this case, the formulas for
the~auto\-correlation and diffraction measures reduce to
\[
  \gamma \, = \, \delta_0 + \bigl(1 - |K|^2(x/p^{1/d})\bigr) \ts \leb^d 
  \quad\text{and}\quad
  \widehat\gamma \, = \,
   \delta_0 + \bigl(1-p \, \widehat{|K|^2}(tp^{1/d})\bigr) \ts \leb^d \ts ,
\]
as can easily be checked.  \eox
\end{remark}

Evidently, the construction below Eq.~\eqref{eq:DPP-ti} gives rise to
a large number of ex\-amples.  Let~us~mention some particularly
interesting cases.

\begin{example}[Sine process]\label{ex:sine}
  An important example is given by the sine process, which corresponds
  to $d = 1$, $K(x) = \tfrac{\sin (\pi x)}{\pi x}$ and $\varphi(t) =
  \pmb{1}_{[-1/2,+1/2]}(t)$.  In~this~case, the autocorrelation and
  diffraction measures are given by
\[
  \gamma \, = \, 
  \delta_0 + \bigl(1 - (\tfrac{\sin (\pi x)}{\pi x})^2 \bigr) \ts \leb
  \quad\text{and}\quad  \widehat\gamma \, = \, 
  \delta_0 + \bigl(1 - \max \{0, 1\!- \! |t| \} \bigr) \ts \leb \ts .
\]
This example arises in connection with the local eigenvalue statistics
of the Gaussian Unitary Ensemble (GUE) in random matrix theory
\cite{AGZ:2010,HKPV:2009,soshnikov:2000}, and is discussed from the
viewpoint of diffraction theory in \cite{baake-koesters:2011}.

By Remark~\ref{rem:thinning}, the sine process gives rise to a whole
family of determinantal point processes, with $K_p(x) = \tfrac{\sin
  (\pi x / p)}{\pi x / p}$ and $\varphi_p(t) = p \,
\pmb{1}_{[-1/(2p),+1/(2p)]}(t)$, where $0 < p \leq 1$.  The
autocorrelation and diffraction measures are now given by
\[
  \gamma_p  \, = \,
   \delta_0 + \bigl(1 - (\tfrac{\sin (\pi x / p)}{\pi x / p})^2\bigr)
   \ts \leb
  \quad\text{and}\quad
  \widehat{\gamma_p} 
  \, = \, \delta_0 + \bigl(1 - p \ts \max \{0, 1 \! - \!\ts p |t| \} 
  \bigr) \ts \leb \,.
\]
Note that for $p > 1$, the function $K_p$ does not give rise to a
determinantal point process, as the condition $0 \leq \varphi_p \leq
1$ is~violated.  Thus, the sine process ($p = 1$) is the point process
with the strongest repulsion in this determinantal family, and
it~seems to be the only member of this family arising in random matrix
theory.  \eox

\begin{figure}
\includegraphics[width=0.48\textwidth]{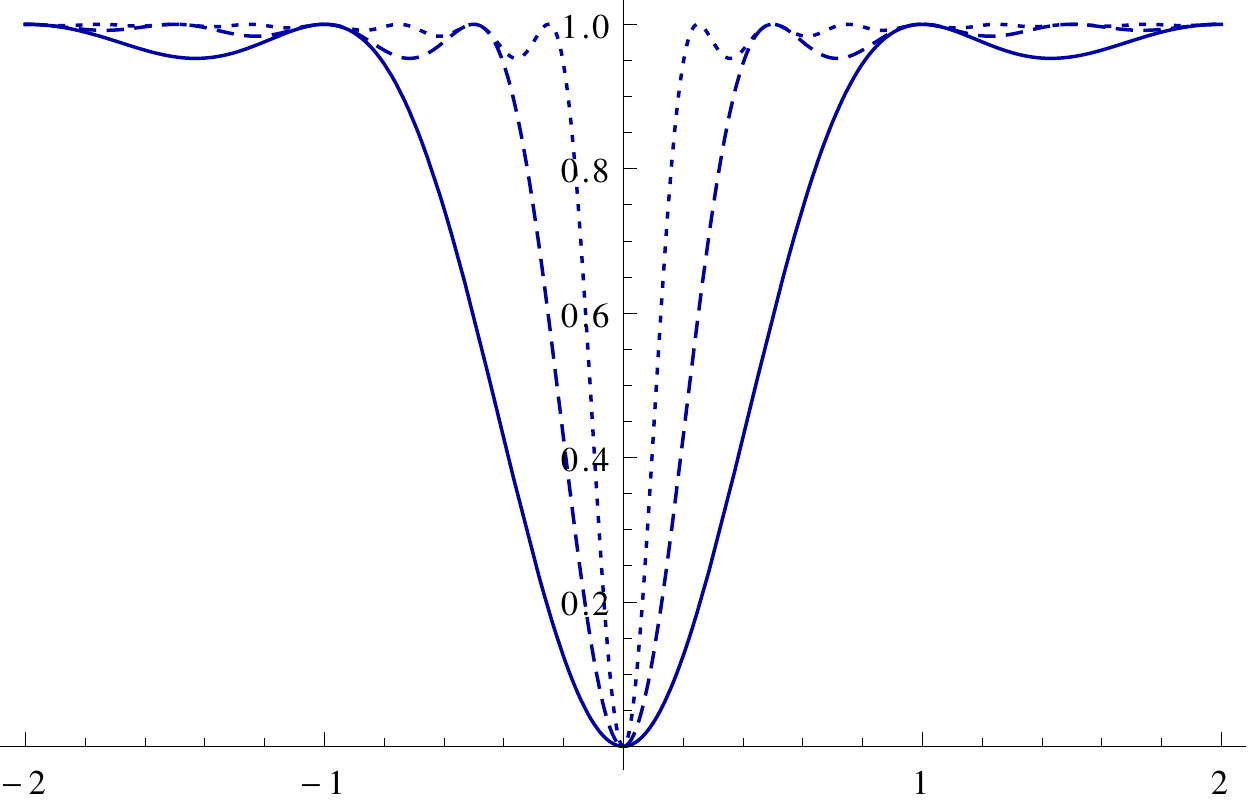}
\quad
\includegraphics[width=0.48\textwidth]{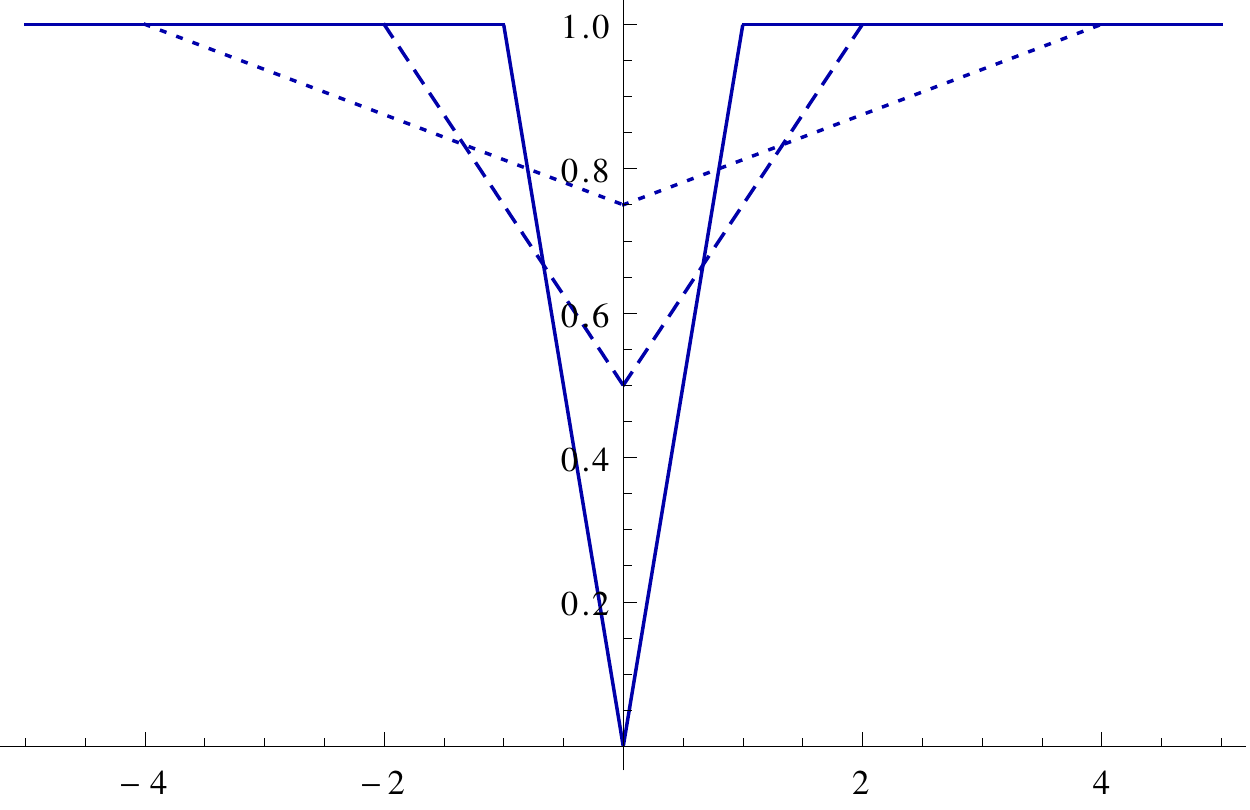}
\caption{The absolutely continuous parts of the autocorrelation (left)
  and diffraction (right) measures of the thinned sine~process for
  $p=1$ (normal), $p=0.5$ (dashed) and $p=0.25$ (dotted).}
\end{figure}
\end{example}

\begin{example}
  Let $d \in \NN$, let $\varphi$ denote the density of the uniform
  distribution on~the $d$-dimensional ball of volume $1$ centered at
  the origin, and let $K := \widehat\varphi$.  Then it is well known
  that $ K(x) = \alpha^{-1/2} |x|^{-d/2} J_{d/2}(2 \pi \alpha^{-1/d}
  |x|), $ where $\alpha := \leb^d(B_1)$ denotes the volume of the
  $d$-dimensional unit ball $B_1$.  Thus, the~auto\-correlation and
  diffraction measures are given by
\[
  \gamma \, = \, 
  \delta_0 + \Bigl( 1 - \alpha^{-1} |x|^{-d} \bigl( J_{d/2}
  (2 \pi \alpha^{-1/d} |x|)\bigr)^2 \Bigr) \ts \leb^d
\]
and
\[
  \widehat\gamma \, = \, 
  \delta_0 + \Bigl( 1 - (\varphi \ast \varphi)(t) \Bigr) \ts \leb^d \ts .
\]
Here $J_{d/2}$ is the Bessel function of the first kind of order $d/2$.
Note that for $d = 1$, we recover the sine process.
\eox
\end{example}

Here is another example of a rotation-invariant kernel.

\begin{example}
  Take $d \in \NN$, $K(x) = e^{-\pi |x|^2}$ and $\varphi(t) = e^{-\pi
    |t|^2}$.  In this case, the auto\-correlation and diffraction
  measures are given by
\[
  \gamma \, = \, 
  \delta_0 + \left( 1 - e^{-2 \pi |x|^2} \right)  \leb^d
  \quad\text{and}\quad
  \widehat\gamma \, = \, 
  \delta_0 + \left( 1 - \left(\tfrac{1}{2}\right)^{d/2} 
  e^{-\pi |t|^2 / 2} \right)  \leb^d \ts ,
\]
by an application of Corollary~\ref{cor:determinantal}.
\eox
\end{example}

Note that the pair $(\gamma,\widehat\gamma)$ comes close to being
self-dual here.  It seems natural to~try to obtain a genuinely
self-dual pair $(\gamma,\widehat\gamma)$ by appropriate rescaling.
However, this would require the transformations 
$e^{-\pi|x|^2} \to e^{-\pi|x|^2/2}$
for the function $K$ and 
$e^{-\pi|t|^2} \to 2^{d/2} e^{-2\pi|t|^2}$
for its Fourier transform $\widehat{K}$, 
and this is not allowed as the~spectrum 
of the corresponding convolution operator
$\mathcal{K}$ is no longer contained in the interval $[0,1]$.

Nevertheless, at least for $d = 2$, there does exist a stationary
determinantal point process with a self-dual pair
$(\gamma,\widehat\gamma)$ of the desired form, although one not coming
from a translation-invariant kernel.

\begin{example}[Ginibre process]\label{ex:ginibre}
  On $\RR^2 \simeq \CC$, consider the kernel
\[
  K(z,w) \, = \, \exp(-\tfrac12 \pi|z|^2 - \tfrac{1}{2} 
   \pi|w|^2+\pi z\overline{w}) \ts .
\]
This kernel is \emph{not} translation-invariant, but one can show that
it still defines a~determinantal point process that is stationary and
ergodic.  By Proposition~\ref{prop:determinantal}, the autocorrelation
and diffraction measures are given by
\[
  \gamma \, = \, \delta_0 + \bigl(1 - e^{-\pi|x|^2}\bigr) \ts \leb^2
  \quad\text{and}\quad
  \widehat\gamma \, = \, \delta_0 + 
  \bigl(1 - e^{-\pi|t|^2} \bigr) \ts \leb^2 \ts .
\]
Note that the pair $(\gamma,\widehat\gamma)$ is self-dual here. 
Note also that the diffraction density vanishes at the origin.  
(Indeed, the integral operator~$\mathcal{K}$ determined by the kernel~$K$ 
is a projection operator here.)  This example arises in connection with 
the local eigenvalue statistics of the Ginibre Ensemble in random matrix
theory \cite{AGZ:2010,HKPV:2009,soshnikov:2000}, and is discussed from
the view\-point of diffraction theory in \cite{baake-koesters:2011}.

Similarly as above, by Remark~\ref{rem:thinning}, we may also consider
thinned versions of the Ginibre process.  Here, the autocorrelation
and diffraction measures are given by
\[
  \gamma \, = \, \delta_0 + \bigl(1 - e^{-\pi |x|^2 /p}\bigr) \ts \leb^2
  \quad\text{and}\quad
  \widehat\gamma \, = \, \delta_0 + 
  \bigl(1 - p e^{-\pi p |t|^2} \bigr) \ts \leb^2 \ts ,
\]
via the usual reasoning. \eox

\begin{figure}
\includegraphics[width=0.48\textwidth]{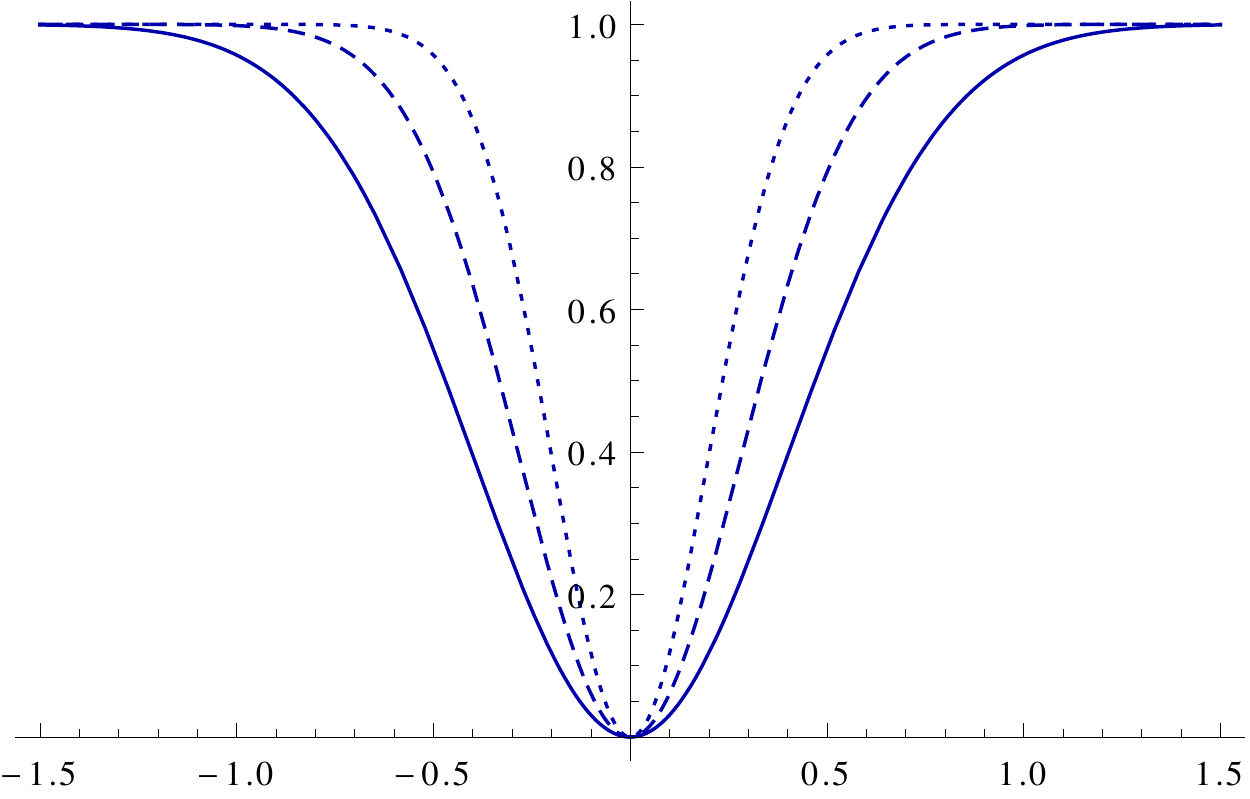}
\quad
\includegraphics[width=0.48\textwidth]{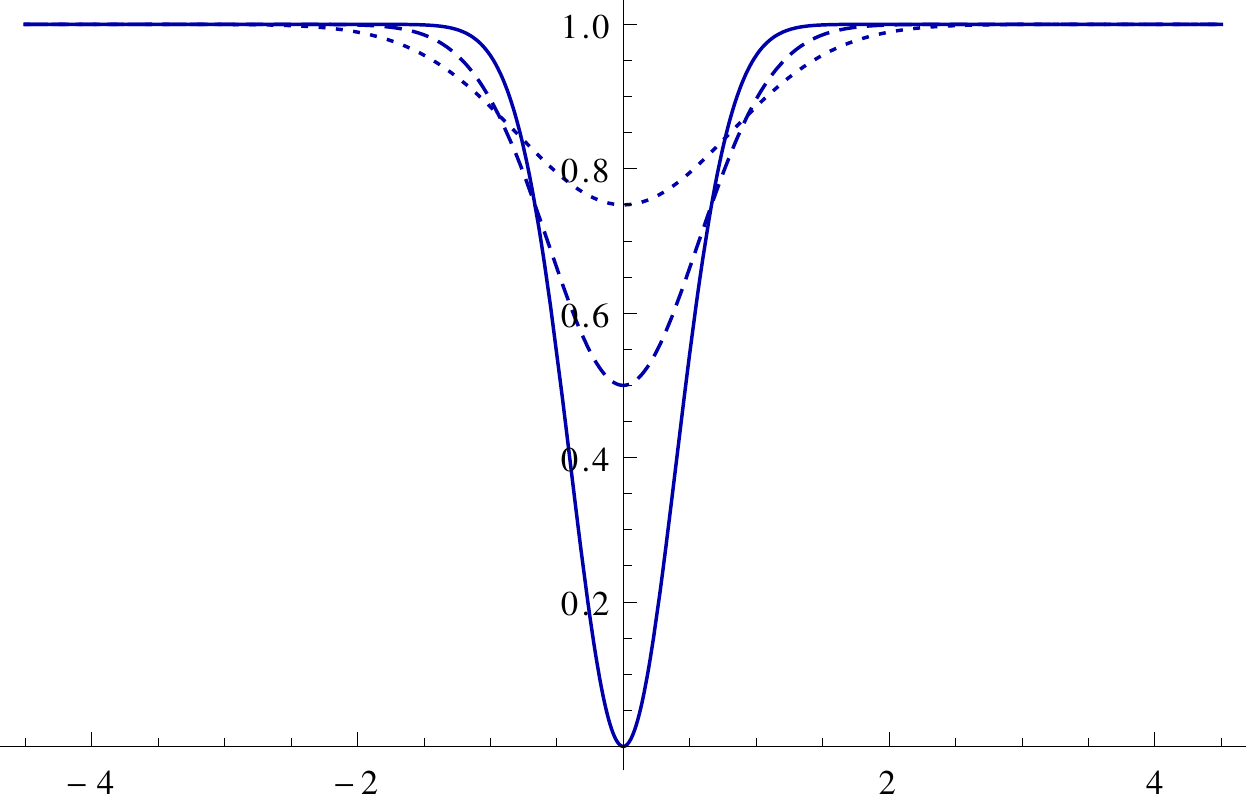}
\caption{The absolutely continuous parts of the autocorrelation (left)
  and diffraction (right) measures (viewed along a line through the
  origin) of the thinned Ginibre~process for $p=1$ (normal), $p=0.5$
  (dashed) and $p=0.25$ (dotted).}
\end{figure}
\end{example}

\begin{example}[Renewal process]
  Another interesting example is given by the class of those
  stationary determinantal point processes which are simultaneously
  renewal processes; see \cite[Section 2.4]{soshnikov:2000} and
  references given there.  Here, $d = 1$, $K(x) = \exp(-|x|/\alpha)$
  and $\varphi(t) = \frac{2\alpha}{1+(2\pi\alpha t)^2}$, where $0 <
  \alpha \leq \tfrac12$.  The density of the increments of the
  associated renewal process is given by
\[
  f_\alpha(x) \, = \, \frac{2}{\sqrt{1-2\alpha}} e^{-x/\alpha} 
  \sinh(\sqrt{1-2\alpha} (x/\alpha)) \, \pmb{1}_{(0,\infty)}(x) \ts ,
\]
see \cite[Eq.~2.42]{soshnikov:2000}.  The autocorrelation and
diffraction measures are~given~by
\[
  \gamma_\alpha \, = \, \delta_0 + 
  \bigl(1 - \exp(-2|x|/\alpha)\bigr) \ts \leb
  \quad\text{and}\quad
  \widehat{\gamma_\alpha} \, = \, \delta_0 + 
  \bigl(1 - \tfrac{\alpha}{1 + \pi^2 \alpha^2 t^2}\bigr) 
  \ts \leb \ts .
\]
Of course, this can also be obtained from the density $f_\alpha$ and
\cite[Theorem 1]{baake-birkner-moody:2010}, which provides formulas
for the autocorrelation and diffraction measures of general renewal
processes.

Similarly to what we saw above, this family of point processes approaches 
the homo\-geneous Poisson process as $\alpha \to 0$, while the kernel 
does not define a determinantal point process for $\alpha > 1/2$.
Note also that for $\alpha = 1/2$ the distribution of the increments is 
the gamma distribution 
with~the density $4x e^{-2x}$
and that all other members of the family can be obtained from the
associated determinantal point process by the thinning procedure
described in Remark~\ref{rem:thinning}.  \eox

\begin{figure}
\includegraphics[width=0.48\textwidth]{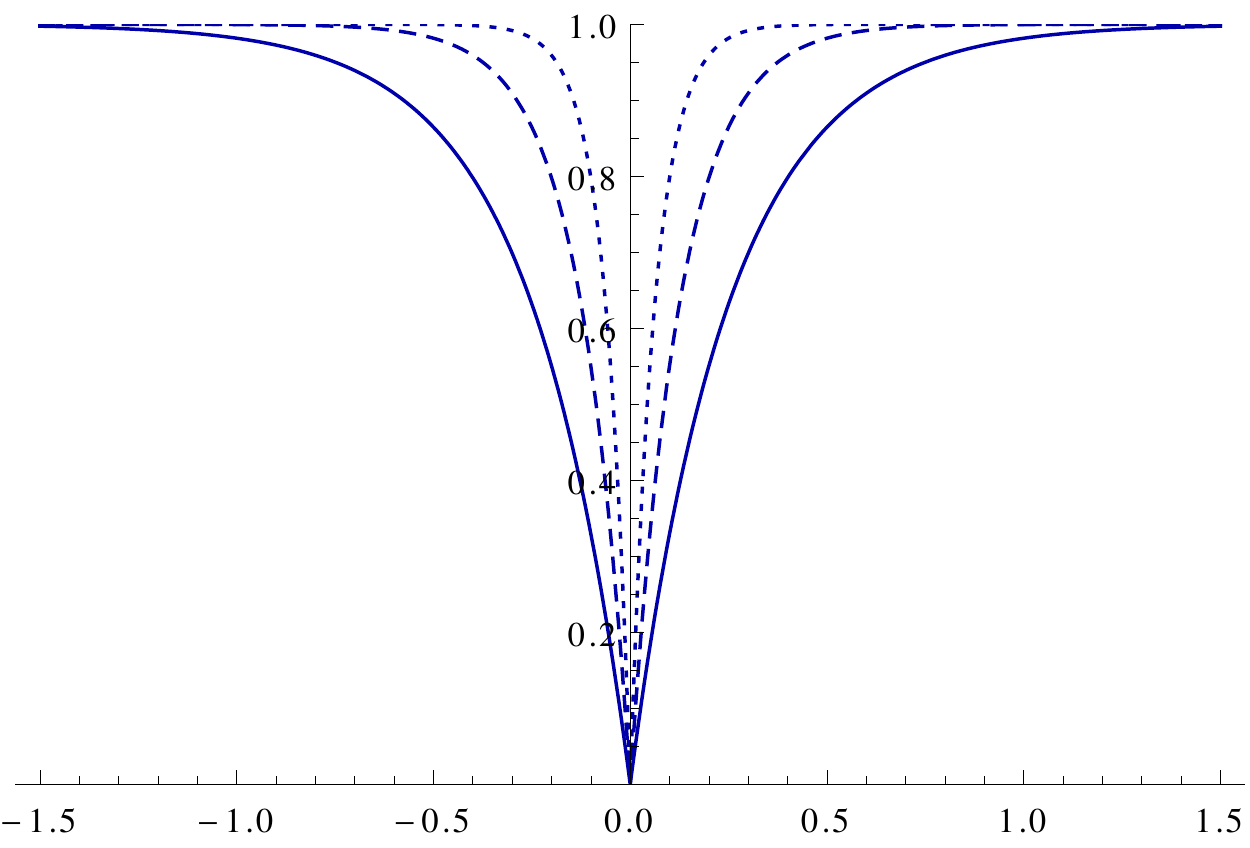}
\quad
\includegraphics[width=0.48\textwidth]{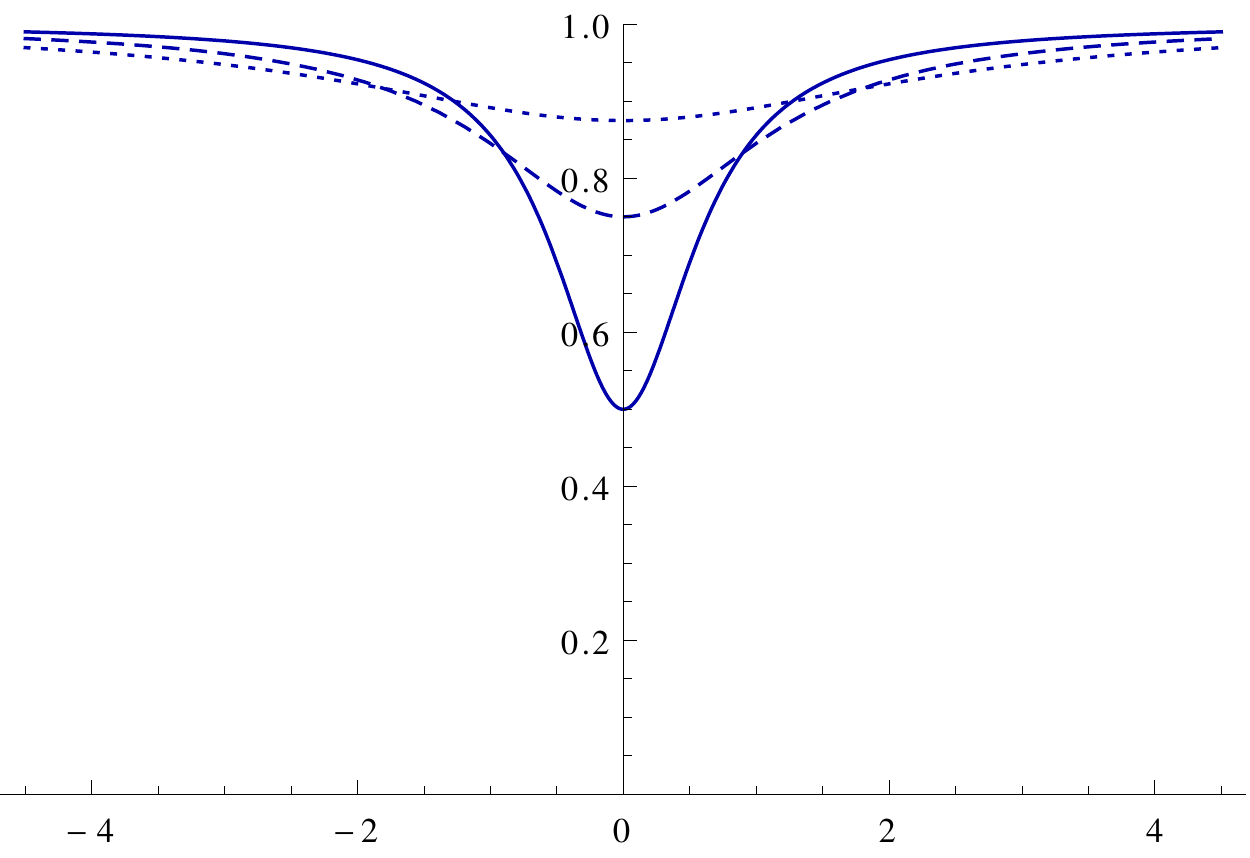}
\caption{The absolutely continuous parts of the autocorrelation (left)
  and diffraction (right) measures of the renewal~process for
  $\alpha=0.5$ (normal), $\alpha=0.25$ (dashed) and $\alpha=0.125$
  (dotted).}
\end{figure}
\end{example}

\begin{example}
  Let $Q_1$ be the Poisson distribution with parameter $1$, let $Q_2$
  be the compound Poisson distribution with parameter $1$ and
  compounding distribution $\tfrac12 \delta_{-1} + \tfrac12
  \delta_{+1}$, and let
\[
  \varphi^{}_1(x) \, := \int_{\RR} \pmb{1}_{[-1/2,+1,2]}(x-y) \d{Q_1}(y)
  \quad\!\text{and}\!\quad
  \varphi^{}_2(x) \, := \int_{\RR} \pmb{1}_{[-1/2,+1,2]}(x-y) \d{Q_2}(y) \,.
\]
Then, $\varphi^{}_1$ and $\varphi^{}_2$ are probability densities
bounded by $1$, and the associated functions $K_1$ and $K_2$ read
\[
   K_1(x) \, = \, \exp\big(e^{-2\pi \i x}-1\big)\, \tfrac{\sin (\pi x)}{\pi x}
  \quad\text{and}\quad
  K_2(x) \, = \, \exp\big(\cos (2\pi x)-1\big)\, \tfrac{\sin (\pi x)}{\pi x} \ts .
\]
Since
\[
  |K_1(x)|^2 \, = \,
   \exp \bigl(2\cos (2\pi x)-2 \bigr) 
   \bigl(\tfrac{\sin \pi x}{\pi x}\bigr)^2 
  \, = \, |K_2(x)|^2 \,,
\]
it follows that the associated stationary determinantal point processes
have the same autocorrelation and diffraction measures.
However, the point processes them\-selves are \emph{not} the same,
as can be verified by comparing the 3-point correlation functions.

Thus, even within the restricted class of determinantal point processes 
with a translation-invariant kernel as in Eq.~\eqref{eq:DPP-ti},
the \emph{inverse problem} to reconstruct the distribution of a point process
from its diffraction measure does not have a~unique solution;
see \cite{baake-grimm:2013} for background information and other examples.
\eox
\end{example}

\begin{remark}[Diffraction spectrum versus dynamical spectrum]
  \label{rem:dynamical-spectrum}
  Let $\omega$ be a stationary and ergodic point process for which the
  first and second moment measures exist.  Then, the diffraction
  measure (or rather its equivalence class) is also called the
  \emph{diffraction spectrum} of $\omega$, whereas the \emph{dynamical
    spectrum} of $\omega$ is the~spectrum of the dynamical system
  defined by the shift operators $T_x$, $x \in \RR^d$, on
  $(\mathcal{N}(\RR^d),\mathscr{N}(\RR^d),\pp_\omega)$.  
  More precisely, the dynamical spectrum may be defined as the maximal
  spectral type of the group of unitary operators $f \mapsto f \circ
  T_x$ on $L^2(\mathcal{N}(\RR^d),\mathscr{N}(\RR^d),\pp_\omega)$; 
  compare \cite{CFS} or \cite{queffelec} for details.
  It is of interest in diffraction theory to clarify 
  the relationship between the diffraction spectrum
  and the dynamical spectrum; see \cite{baake-lenz-vanenter:2013}
  and references therein for background information.

  Let us consider the diffraction spectrum and the dynamical spectrum
  for a determinantal point process with a translation-invariant
  kernel as in Eq.~\eqref{eq:DPP-ti}.  On the one hand, as we have
  seen above, the diffraction spectrum is equivalent to $\delta_0 +
  \leb^d$.  On the other hand, the~dynamical spectrum is also
  equivalent to $\delta_0 + \leb^d$.  Indeed, it was shown
  in~\cite{soshnikov:2000} that the determinantal point process is
  absolutely continuous (when viewed as a dynamical system), so that
  the dynamical spectrum is dominated by $\delta_0 + \leb^d$.
  Furthermore, as also shown in \cite{soshnikov:2000}, the centred
  linear statistics $\omega \mapsto \int f(x) \d\omega(x) - \ee(\int
  f(x) \d\omega(x))$, where $f \in \mathscr{C}_{c}^{\infty}(\RR^d)$,
  possess the spectral measure $(1-\widehat{|K|^2}(t)) \,
  |\widehat{f}(t)|^2 \, \lambda^d$, which implies that the dynamical
  spectrum must be equivalent to $\delta_0 + \leb^d$.  Thus, 
  the `diffraction spectrum' and the `dynamical spectrum' are equivalent here.  
  
  It seems interesting to ask whether this is a general property
  of translation-invariant point processes with ``good'' mixing properties.
  Let us mention here that the determinantal point process is not only mixing, 
  but even mixing of all orders; see~\cite{soshnikov:2000}.
  \eox
\end{remark}

\section{Permanental point processes}
\label{sec:PPP}

We now turn to \emph{permanental point processes} where the
correlation functions are given by the permanent instead of the
determinant of a certain kernel.  For such processes, the particles
tend to form clumps, whereas they repel one another for determinantal
point processes.  See \cite[Section~4.9]{HKPV:2009} for background
information.  See also \cite{macchi:1975} and \cite[Example
6.2\,(b)]{DV1}, where these point processes are called \emph{boson
  processes}.

Let $K$ be a kernel as in Eq.~\eqref{eq:main}.  A point process
$\omega$ on $\RR^d$ is called \emph{permanental} with kernel $K$ 
if, for any $k \in \NN$, the $k$-point correlation function exists 
and is given by
\begin{equation}\label{eq:PPP-def}
   \CF{k}(x_1,\hdots,x_k) \, = \, 
   \per \bigl( K(x_i,x_j)^{}_{1 \leq i,j \leq k} \bigr) ,
\end{equation}
where $\per$ denotes the permanent; 
compare \cite[Definition 2.1.5]{HKPV:2009}.

We refer to \cite[Section~4.9]{HKPV:2009} for the proof of the
following result.

\begin{proposition}[{\cite[Corollary
    4.9.3]{HKPV:2009}}]\label{prop:PPP-existence}
  For any kernel\/ $K$ as in Eq.~\eqref{eq:main}, there exists a
  permanental point process with kernel\/ $K$ on\/ $\RR^d$.  \qed
\end{proposition}

Furthermore, it is not hard to see that if there exists a permanental
point process $\omega$ with a given kernel $K$ as in
Eq.~\eqref{eq:main}, its distribution is uniquely determined.  This
follows from the observation that, for any bounded Borel set $A$, 
the probability generating function $z \mapsto \ee(z^{\omega(A)})$ 
exists in an open neighborhood of the unit ball; 
see \cite[Theorem 6]{macchi:1975}.

\begin{remark}
  Let us note here that the probability generating function $z \mapsto
  \ee(z^{\omega(A)})$ is an entire function for determinantal, but generally 
  not for permanental point processes.  The reason for this difference is
  that determinants satisfy Hadamard's inequality, while there is no
  comparable estimate for permanents.  \eox
\end{remark}

Of course, we will be interested in permanental point processes which
are also stationary with mean density $1$.  For brevity, let us
directly turn to permanental point processes with
translation-invariant kernels.  More precisely, we shall assume the
following:
\begin{align}
\label{eq:PPP-ti}
\begin{minipage}{0.8\textwidth}
  $K(x,y) = K(x-y)^{}_{\vphantom{I}}$ holds for all $x,y \in \RR^d$,
  where the function $K \! : \, \RR^d \to \CC$ on the
  right-hand side is the Fourier transform of a~probability density
  $\varphi$ on $\RR^d$.
\end{minipage}
\end{align}

Note that, in contrast to Condition \eqref{eq:DPP-ti}, the probability
density need not be bounded here.  However, Condition
\eqref{eq:PPP-ti} still implies Condition \eqref{eq:main}.

Now suppose that Condition \eqref{eq:PPP-ti} holds.  Let $\omega$ 
denote the associated permanental point process, which exists by
Proposition \ref{prop:PPP-existence}.  Since the correlation functions
determine the dis\-tribution of $\omega$, $\omega$ is clearly
stationary with mean density $K(0) = 1$.  Moreover, by the
Riemann--Lebesgue lemma, we have $K(x) \longrightarrow 0$ as $|x| \to
\infty$, and~a~variation of the proof of
\cite[Theorem~7]{soshnikov:2000} or \cite[Theorem 4.2.34]{AGZ:2010}
shows that $\omega$ is also mixing and hence ergodic.

Furthermore, by Eq.~\eqref{eq:translation-invariance}, (the continuous
versions of) the first and second correlation functions of $\omega$
satisfy
\begin{equation}\label{eq:PPP-first-and-second}
   \CF{1}(x_1) \, = \, K(0) \, = \, 1 
   \quad\text{and}\quad
   \CF{2}(x_1,x_2) \, = \, 1 + \bigl|K(x_1-x_2)\bigr|^2 \ts .
\end{equation}
Therefore, similarly as in Corollary~\ref{cor:ET}, 
we have the following result:

\begin{proposition}\label{prop:permanental}
  Let $K$ be a kernel as in Eq.\@ \eqref{eq:PPP-ti}.  Then, the
  auto\-correlation and diffraction measures of the associated
  permanental point process $\omega$ are given by
\[
   \gamma \, = \, \delta_0 + \bigl(1 + |K|^2 \bigr) \ts \leb^d 
   \quad\text{and}\quad
   \widehat\gamma \, = \, \delta_0 + \bigl( 1+
    (\varphi \ast \varphi^{}_{-}) \bigr) \ts \leb^d \ts ,
\]
where $\varphi^{}_{-}(x) := \varphi(-x)$ as before.
\end{proposition}

Note that the autocorrelation and diffraction densities are larger
than $1$ here, in~line with the clumping picture.  Also, under the
above-mentioned assumptions, the diffraction measure is absolutely
continuous apart from the Bragg peak at the~origin, and the absolutely
continuous part of the diffraction measure is~equivalent to Lebesgue
measure.

Note that if $K$ were square-integrable, we could also write
$\widehat{|K|^2}$ instead of $(\varphi \ast \varphi^{}_{-})$ in the
result for the diffraction measure, similarly as in Corollary
\ref{cor:determinantal}.  However, in~general, $K$ need not be
square-integrable here.

\begin{proof}[Proof of Proposition $\ref{prop:permanental}$]
  This follows from the proof of Corollary~\ref{cor:ET} and
  the observation that $(|K|^2 \, \leb)\widehat\enspace = 
  (\varphi \ast \varphi^{}_{-}) \, \leb$.

  To check this observation, note that $K = \widehat\varphi$ implies
  $|K|^2 = (\varphi \ast \varphi^{}_{-})\widehat\enspace$.  Thus,
  $((\varphi \ast \varphi^{}_{-}) \, \leb)\widehat\enspace = |K|^2 \,
  \leb$, and the desired relation follows by Fourier inversion in the
  space of positive and positive-definite measures, and the fact that
  $(\varphi \ast \varphi^{}_{-}) \, \leb$ is a~symmetric measure.
\end{proof}

Clearly, all examples for determinantal point processes translate into
examples for permanental point processes.  In particular, the thinning
procedure described in Remark~\ref{rem:thinning} also extends to
permanental point processes.  However, the kernel $K_p$ introduced
there may now be considered also for $p > 1$, where it still defines a
permanental point process $\omega_p$.  (Of course, the probabilistic
description in terms of thinning breaks down in this region.  However,
at least for natural numbers $p$, $\omega_p$ may be viewed as the
superposition of $p$ independent copies of $\omega$.  In fact, this is
not surprising in view of the representation as a Cox process to be
mentioned below.)  Thus, any permanental point process $\omega$ gives
rise to a whole family of permanental point processes $(\omega_p)_{0 <
  p < \infty}$.  These families interpolate between the homogeneous
Poisson process (for~$p \to 0$) and the \emph{non-ergodic} mixed
Poisson process with the exponential distribution with parameter $1$
as mixing distribution (for~$p \to \infty$).

\begin{remark}
A careful analysis of the arguments in \cite{soshnikov:2000}
shows that, with the~obvious modifications, Remark
\ref{rem:dynamical-spectrum} continues to hold for permanental point
processes as in Proposition~\ref{prop:permanental}.  In particular,
both the diffraction spectrum and the dynamical spectrum are
equivalent to $\delta_0 + \leb^d$ here.
\end{remark}

\begin{remark}
  As mentioned in the last section, our Condition \eqref{eq:DPP-ti} in
  the investigation of determinantal point processes is essential
  in the sense that it must be satisfied for any determinantal point process 
  with a \emph{translation-invariant} kernel satisfying Eq.~\eqref{eq:main}.  
  In contrast, our Condition \eqref{eq:PPP-ti} in the investigation 
  of permanental point processes could be relaxed.

  For instance, we could start from the assumption that the function
  $K : \RR^d \to \CC$ in Eq.~\eqref{eq:PPP-ti} is the Fourier
  transform of a continuous (but not absolutely continuous)
  probability measure $Q$ on $\RR^d$.  Let $\omega$ be the associated
  permanental point process, which exists by
  Prop.~\ref{prop:PPP-existence}.  Then, it is not necessarily true
  that $K(x) \longrightarrow 0$ as $|x| \to \infty$, but one can
  convince oneself that $\omega$ is still stationary and ergodic;
  see Section~5 for details.  Hence, an argument similar to
  the proof of Prop.~\ref{prop:permanental} leads to the conclusion that
  the auto\-correlation and diffraction measures of $\omega$
  are~given~by
\[
  \gamma \, = \, \delta_0 + \bigl(1 + |K|^2 \bigr) \ts \leb^d
  \quad\text{and}\quad
  \widehat\gamma \, = \, \delta_0 + \leb^d + 
  \bigl(Q \ast Q^{}_{-} \bigr) \ts ,
\]
where $Q^{}_{-}(A) := Q(-A)$ denotes the reflection of $Q$ at the
origin.  Of course, the diffraction measure may now contain a
singular continuous component.

Note that we required the probability measure $Q$ to be continuous.
If the prob\-ability measure $Q$ is not continuous, the corresponding
Fourier transform $K$ still gives rise to a stationary permanental
point process $\omega$, but this point process need not be ergodic
anymore.  A simple (counter)example is given by the kernel $K \equiv
1$, for which the associated permanental point process is a mixed
Poisson process with directing measure $Z \, \leb$, where $Z$ has an
exponential distribution with parameter $1$.  This point process is
stationary but not ergodic, and the autocorrelation measure is equal
to $\gamma = Z \, \delta_0 + Z^2 \, \leb$, thus depends on 
the realisation.  \eox
\end{remark}

\section{Cox processes}
\label{sec:COX}

Recall the definition of a Cox process from \cite[Section~6.2]{DV1}.
Let $\omega_0$ be a random measure on $\RR^d$.  A point process
$\omega$ on $\RR^d$ is called \emph{Cox process directed by $\omega_0$} if, 
conditionally on $\omega_0$ (i.e.\@ when $\omega_0$ is regarded as fixed),
$\omega$ is a Poisson process with intensity measure $\omega_0$.

It is a standard result that a Cox process on $\RR^d$ is simple if and only
if the directing measure is continuous.  Furthermore, it is well known
that a Cox process on $\RR^d$ is stationary [ergodic, mixing] if and
only if the directing measure is stationary [ergodic, mixing]; compare
\cite[Proposition~12.3.7]{DV2}.

For the formulation of the next result, let us recall that, for a
general random measure $\omega$, the $k$th moment measure $\MM{k}$ is
defined as the expectation measure (if it exists) of the product
measure $\omega^k$, and for a stationary random measure $\omega$, the
$k$th reduced moment measure $\RMM{k}$ is then defined
similarly as in Section~2.  Furthermore, for a stationary random
measure with mean density $1$, we may define the \emph{reduced
  covariance measure} by
\begin{equation}\label{eq:cumulant-measure-1}
   \RCM{2} \, = \, \RMM{2} - \leb^d \,.
\end{equation}
For a stationary point process with mean density $1$, we may
additionally define the \emph{reduced factorial covariance measure} by
\begin{equation}\label{eq:cumulant-measure-2}
   \RFCM{2} \, = \, \RFMM{2} - \leb^d \,.
\end{equation}

From Theorem~\ref{thm:ET}, we~obtain the following result.
\begin{proposition}\label{prop:cox}
  Let\/ $\omega$ be a stationary and ergodic Cox process 
  with a directing measure\/ $\omega_0$ for which
  the first and second moment measures exist,
  and suppose that\/ $\omega_0$, and hence
  $\omega$, have mean density $1$.  Then, almost~surely, the
  autocorrelation and diffraction measures of\/ $\omega$ are given by
\[
  \gamma \, = \, \delta_0 + \leb^d + \kappa^{}_0
  \quad\text{and}\quad
  \widehat\gamma \, = \, 
  \delta_0 + \leb^d + \widehat{\kappa^{}_0} \ts ,
\]
where\/ $\kappa^{}_0$ is the reduced covariance measure of\/
$\omega_0$, and\/ $\widehat{\kappa^{}_0}$ its Fourier transform.
\end{proposition}

Let us note that the reduced covariance measure $\kappa^{}_0$ is a
positive-definite measure, so that the Fourier transform
$\widehat{\kappa^{}_0}$ exists as a positive measure.  Also, let us
note that $\widehat{\kappa^{}_0}$ is also known as the \emph{Bartlett
  spectrum} of $\omega_0$ in the literature; see \eg \cite[Chapters
8.1 and 8.2]{DV1} for more information.  (More precisely, the Bartlett
spectrum is defined as the inverse Fourier transform of the reduced
covariance measure.  However, as the reduced covariance measure is
symmetric, the Fourier transform and the inverse Fourier transform
coincide, at least for our definition of the Fourier transform.)

\begin{proof}[Proof of Proposition $\ref{prop:cox}$]
  By \cite[Proposition 6.2.2]{DV1}, the reduced factorial covariance
  measure of the Cox process $\omega$ equals the reduced covariance
  measure of the directing measure $\omega_0$,
  \ie $\RFCM{2}(\omega) = \RCM{2}(\omega_0)$.
  Here, the measures in the brackets indicate 
  which measures the moment measures belong to.  
  It~therefore follows
  from Theorem~\ref{thm:ET} and Eq.~\eqref{eq:cumulant-measure-2} that
  the auto\-correlation of $\omega$ is given by
\[
  \gamma \, = \, \delta_0 + \RFMM{2}(\omega) 
  \, = \, \delta_0 + \leb^d + \RFCM{2}(\omega)
  \, = \, \delta_0 + \leb^d + \RCM{2}(\omega_0) \ts . 
\]
Taking the Fourier transform completes the proof.
\end{proof}

\begin{remark}
  It is well known that permanental point processes are special cases
  of Cox processes; see \cite{macchi:1975} or \cite[Proposition
  4.9.2]{HKPV:2009}.  For the convenience of the reader, and since we
  can use this connection to establish the ergodicity of permanental
  point processes, let us outline the argument in a simple situation.

  Suppose that the kernel $K$ is \emph{translation-invariant} and
  satisfies Eq.~\eqref{eq:main} and $K(0,0)$ $= 1$.  Then, the
  underlying function $K : \RR^d \to \CC$ is continuous, Hermitian,
  and positive-definite with $K(0) = 1$, and hence the covariance
  function of a stationary complex Gaussian process $(X_t)_{t \in
    \RR^d}$.  To avoid technical issues, let us assume that $(X_t)_{t
    \in \RR^d}$ has continuous sample paths.  Then, it is not
  difficult to check (using Wick's formula for the moments of complex
  Gaussian random variables; see~\eg \cite[Lemma 2.1.7]{HKPV:2009})
  that the Cox process $\omega$ directed by $\omega_0 := |X_t|^2 \,
  \leb^d$ is a~permanental point process with kernel $K$.

  Furthermore, if $(X_t)_{t \in \RR^d}$ is stationary [ergodic,
  mixing], then $\omega_0$ is also stationary [ergodic, mixing], being
  a factor in the sense of ergodic theory, and this implies that
  $\omega$ is stationary [ergodic, mixing] by the above-mentioned
  results on Cox processes.  Thus, we~obtain useful sufficient conditions 
  for ergodicity and mixing of Cox processes from the well-known theory of
  stationary Gaussian processes: $\omega$ is ergodic if the~spectral
  measure of $(X_t)_{t \in \RR^d}$ is continuous, and $\omega$ is
  mixing if $K(x) \longrightarrow 0$ as $|x| \to \infty$.  \eox
\end{remark}

Let us end this section with an example demonstrating that stationary
and ergodic Cox processes can have additional Bragg peaks apart from
the origin.

\begin{example}
  Consider the continuous stochastic process $X = (X_t)_{t \in \RR}$
  with $X_t = \linebreak 1 + \cos(2\pi(t + U))$, where $U$ is
  uniformly distributed on $[0,1]$.  Let $\omega_0 := X_t \, \leb$ be
  the random measure with density $X_t$, and let $\omega$ be the Cox
  process directed by $\omega_0$.  Then, one can check that
  $\omega_0$, and hence $\omega$, is stationary and ergodic.
  Furthermore, it is easy to check that the reduced covariance measure
  of $\omega_0$ is given by $\kappa^{}_0 = \tfrac12 \cos(2 \pi x) \,
  \leb$.  It therefore follows from Proposition~\ref{prop:cox} that
  the autocorrelation and diffraction measures of $\omega$ are given
  by
\[
  \gamma \, = \, \delta_0 + \lambda + \tfrac12 \cos (2 \pi x) \ts \leb
  \quad\text{and}\quad
  \widehat\gamma \, = \, \delta_0 + \lambda + \tfrac{1}{4} 
  \bigl(\delta_{-1} + \delta_{+1}\bigr)  ,
\]
respectively.  \eox
\end{example}

\section{Zeros of Gaussian random analytic functions}
\label{sec:GAF}

A \emph{Gaussian random analytic function} is a random variable $f$
whose values are analytic functions on $\CC$ with the property that, 
for all $n \in \NN$ and for all choices of $z_1,\hdots,z_n \in \CC$,
the~$n$-tuple $(f(z_1),\hdots,f(z_n))$ has a complex Gaussian 
distribution with mean $0$. One such example is the Gaussian
random analytic function $f$ given by
\begin{equation}\label{eq:GAF}
  f(z) \, := \sum_{n=0}^{\infty} a^{}_n 
  \frac{\sqrt{L^n}}{\sqrt{n!}} z^n \,,
\end{equation}
where $L$ is a positive constant and the $a_n$
are \iid standard complex Gaussian random variables.  We are
interested in the zero set of $f$ viewed as a point process on
$\CC \simeq \RR^2$.  By \cite[Proposition~2.3.4]{HKPV:2009}, 
the distribution of the zero set of $f$ is invariant under translations 
(and also under rotations), and by~\cite[Corollary~2.5.4]{HKPV:2009}, 
$f$ is essentially the only Gaussian random analytic function with 
this property.  
Furthermore, by the proof of \cite[Proposition~2.3.7]{HKPV:2009}, 
the~zero set of $f$ defines an ergodic point process with respect to 
the group of translations. 
Indeed, the zero set of $f$ is even mixing:

\begin{proposition}
\label{prop:GAF-mixing}
The point process given by the zero set of the Gaussian random 
analytic function $f$ in Eq.~\eqref{eq:GAF} is mixing.
\end{proposition}

\begin{proof}
We use a similar argument as in the proof of \cite[Proposition~2.3.7]{HKPV:2009}.
  For convenience, let us suppose that $L = 1$.
  Then, using that the covariance kernel 
  of the complex Gaussian process $f$
  is given by $K(z,w) = \exp(z \overline{w})$
  (cf.\@ Equation \eqref{eq:covariance-1} below),
  it is straightforward to check that, for any $\zeta \in \CC$,
  the complex Gaussian processes
  $(f(z+\zeta) e^{-|z+\zeta|^2/2} e^{-\i \im(z \overline\zeta)})_{z \in \CC}$
  and
  $(f(z) e^{-|z|^2/2})_{z \in \CC}$
  have the same distribution.
  As a consequence, the stochastic process $(v(z))_{z \in \CC}$
  given by
\[
v(z) := |f(z)| e^{-|z|^2/2}
\]
  is stationary, \ie for any $\zeta \in \CC$,
  $(v(z+\zeta))_{z \in \CC}$ and $(v(z))_{z \in \CC}$
  have the same distribution.
  Furthermore, the stochastic process $(v(z))_{z \in \CC}$ is mixing,
  \ie for any events  $A,B \in \mathcal{B}(\mathscr{C}(\CC))$, \
\begin{align}
\label{eq:mixing}
  \pp\big((v(z+\zeta))_{z \in \CC} \in A \,\wedge\, (v(z))_{z \in \CC} \in B\big)
  \xrightarrow{\,|\zeta| \to \infty\,}
  \pp\big((v(z))_{z \in \CC} \in A\big) \, \pp\big((v(z))_{z \in \CC} \in B\big) \ts .
\end{align}
  Here $\mathscr{C}(\CC)$ denotes the space 
  of continuous functions $\varphi : \CC \to \RR$
  (endowed with the topology of locally uniform convergence),
  and $\mathcal{B}(\mathscr{C}(\CC))$ denotes its Borel $\sigma$-field,
  which~coincides with the Borel $\sigma$-field
  generated by the projections $\pi_z$, with $z \in \CC$.
  By standard arguments, it suffices to~check \eqref{eq:mixing}
  for events $A$ and $B$ of the form 
  $A = \bigcap_{j=1}^{m} \pi_{z_j}^{-1}(A_j)$
  and
  $B = \bigcap_{k=1}^{n} \pi_{w_k}^{-1}(B_k)$,
  where $m,n \in \NN$, $z_j,w_k \in \CC$, and $A_j,B_k \subset \RR$ are Borel sets.
  But now, again using that the covariance kernel of $f$
  is given by $K(z,w) = \exp(z \overline{w})$,
  it is easy to see that the random vectors
  $(f(z_j+\zeta) e^{-|z_j+\zeta|^2/2} e^{- \i \im(z_j \overline\zeta)})_{j=1,\hdots,m}$
  and
  $(f(w_k) e^{-|w_k|^2/2})_{k=1,\hdots,n}$ 
  are asymptotically independent as $|\zeta| \to \infty$.
  Therefore,
\begin{multline*}
  \pp\big(|f(z_j+\zeta)| e^{-|z_j+\zeta|^2/2} \in A_j \ \forall j \,\wedge\, |f(w_k)| e^{-|w_k|^2/2} \in B_k \ \forall k\big) \\
  \xrightarrow{\,|\zeta| \to \infty\,}
  \pp\big(|f(z_j)| e^{-|z_j|^2/2} \in A_j \ \forall j\big) \, \pp\big(|f(w_k)| e^{-|w_k|^2/2} \in B_k \ \forall k\big) \ts ,
\end{multline*}
  and \eqref{eq:mixing} is proved.
  Since the zero set of the Gaussian random analytic function $f$ 
  may be represented as a factor (in the sense of ergodic theory)
  of the stochastic process $(v(z))_{z \in \CC}$, 
  this establishes Proposition \ref{prop:GAF-mixing}.
\end{proof}

Since the point process of zeros is ergodic
and the moment measures of any order exist (see~below for details), 
the autocorrelation and diffraction measures exist by Theorem \ref{thm:ET}.
By \cite[Corollary~3.4.2]{HKPV:2009}, the $k$-point correlation functions 
of the zero set of $f$ are given by
\begin{equation}\label{eq:gaf-correlationfunction}
   \CF{k}(z^{}_1,\hdots,z^{}_k) \, = \, 
   \per(C - B A^{-1} B^*) / \det(\pi A) \ts ,
\end{equation}
where the $k \times k$ matrices $A,B,C$ have the entries
\[
  A(i,j) \, := \, \ee(f(z_i) \overline{f(z_j)}) \, ,\quad
  B(i,j) \, := \, \ee(f'(z_i) \overline{f(z_j)}) \, ,\quad
  C(i,j) \, := \, \ee(f'(z_i) \overline{f'(z_j)}) \, ,
\]
and $B^*$ denotes the conjugate transpose of $B$.  
Straightforward calculations yield
\begin{align}
\label{eq:covariance-1}
  \ee(f(z) \overline{f(w)}) \, = 
  \sum_{n=0}^{\infty} \ee \bigl(|a_n|^2\bigr) 
  \frac{L^n}{n!} z^n \,\overline{w}^n 
  \, = \, \exp(L z \ts \overline{w}) \, ,
\end{align}
\begin{align}
\label{eq:covariance-2}
  \ee(f'(z) \overline{f(w)}) \, =
  \sum_{n=0}^{\infty} \ee \bigl(|a_n|^2 \bigr) 
  \frac{L^n}{n!} \, n\ts z^{n-1} \ts \overline{w}^{n} 
  \, = \, L \ts \overline{w} \, \exp(L z \ts \overline{w}) \, ,
\end{align}
\begin{align}
\label{eq:covariance-3}
  \ee(f'(z) \overline{f'(w)}) \, = 
  \sum_{n=0}^{\infty} \ee \bigl(|a_n|^2 \bigr) 
  \frac{L^n}{n!} \, n \ts z^{n-1} n \ts \overline{w}^{n-1}
  \, = \, (L^2 z\ts \overline{w} + L)\, \exp(L z \overline{w}) \, . 
\end{align}
If we insert this into Eq.~\eqref{eq:gaf-correlationfunction} (for
$k=1$ and $k=2$), we find after some calculation that
\[
   \CF{1}(z) \, = \, \frac{L}{\pi}
\]
and
\begin{align*}
  \CF{2}(z_1,z_2) \, = \, 
  & \frac{L^2 \exp \bigl(L|z_1-z_2|^2 \bigr) 
  \Big( 1 - \exp(L|z_1-z_2|^2) + L|z_1-z_2|^2 \Big)^2}
  {\pi^2 \Big( \exp(L|z_1-z_2|^2) - 1 \Big)^3} \\ 
  & + \frac{L^2 \Big( 1 - \exp \bigl(L|z_1-z_2|^2 \bigr) 
  + L|z_1-z_2|^2 \exp \bigl(L|z_1-z_2|^2 \bigr) \Big)^2}
  {\pi^2 \Big( \exp \bigl(L|z_1-z_2|^2 \bigr) - 1 \Big)^3} \ts .
\end{align*}
In particular, the $2$-point correlation function depends on $z_1$ and
$z_2$ only via their distance $r := |z_1-z_2|$, as it should.  From
now on, we will always set $L = \pi$, so~that the mean density of the
point process is equal to $1$.  Then, expressing the two-point
correlation function in terms of $r$, we~obtain
\[
  \CF{2}(0,r) \, = \, 1 - g(r) \ts ,
\]
where 
\begin{equation}\label{eq:function-g}
  g(r) \, := \, \frac{e^{-\pi r^2} 
  \big( -2+4\pi r^2-\pi^2 r^4 \big) + 
  e^{-2\pi r^2} \big( 4-4\pi r^2-\pi^2 r^4 \big) 
  - 2 e^{-3\pi r^2}}{\big( 1 - e^{- \pi r^2} \big)^3} \ts . 
\end{equation}
Moreover, an explicit calculation with the Fourier transform of a
radially symmetric function shows that the Fourier transform of $g$ is
given by
\begin{equation}\label{eq:function-h}
  h(s) \, := \, 1 + \sum_{k=2}^{\infty} 
  \frac{(-1)^{k+1}}{(k-2)!} \pi^k s^{2k} \zeta(k+1) \ts ;
\end{equation}
see below for details.
Therefore, Corollary~\ref{cor:ET} implies the following result.

\begin{theorem}\label{thm:GAF}
  Let\/ $\omega$ be the point process given by the zeros of the
  Gaussian random analytic function in Eq.~\eqref{eq:GAF}, 
  with\/ $L = \pi$. Then, the autocorrelation and diffraction measures 
  of\/ $\omega$ are given by
\[
   \gamma \, = \, \delta_0 + \bigl(1 - g(r) \bigr) \ts \leb^2
\quad\text{and}\quad
   \widehat{\gamma} \, = \, \delta_0 + \bigl(1 - h(s) \bigr) \ts \leb^2 \ts ,
\]
where $r \equiv r(x_1,x_2) := \sqrt{x_1^2 + x_2^2}$, $s \equiv
s(t_1,t_2) := \sqrt{t_1^2 + t_2^2}$, and $g(r)$ and $h(s)$ are the
functions defined in Eqs.~\eqref{eq:function-g} and
\eqref{eq:function-h}.
\end{theorem}

\begin{figure}
\includegraphics[width=0.48\textwidth]{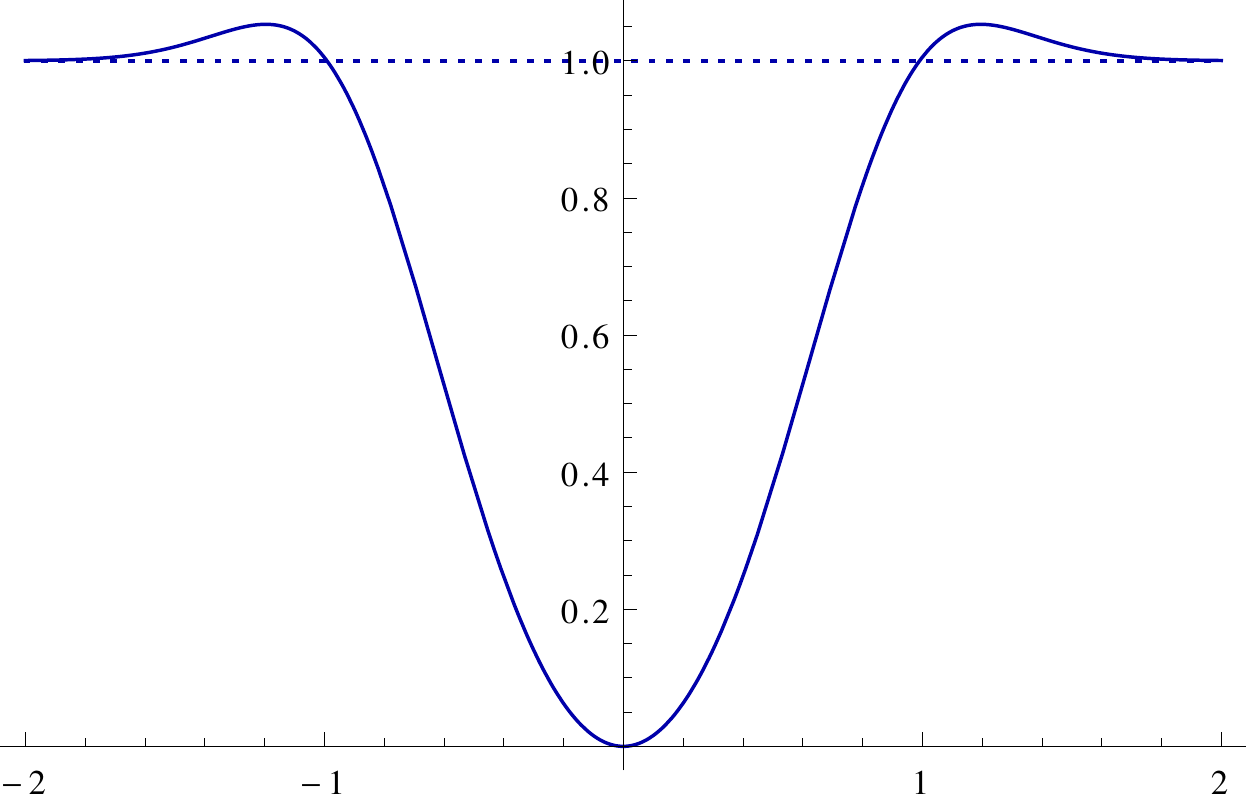}
\includegraphics[width=0.48\textwidth]{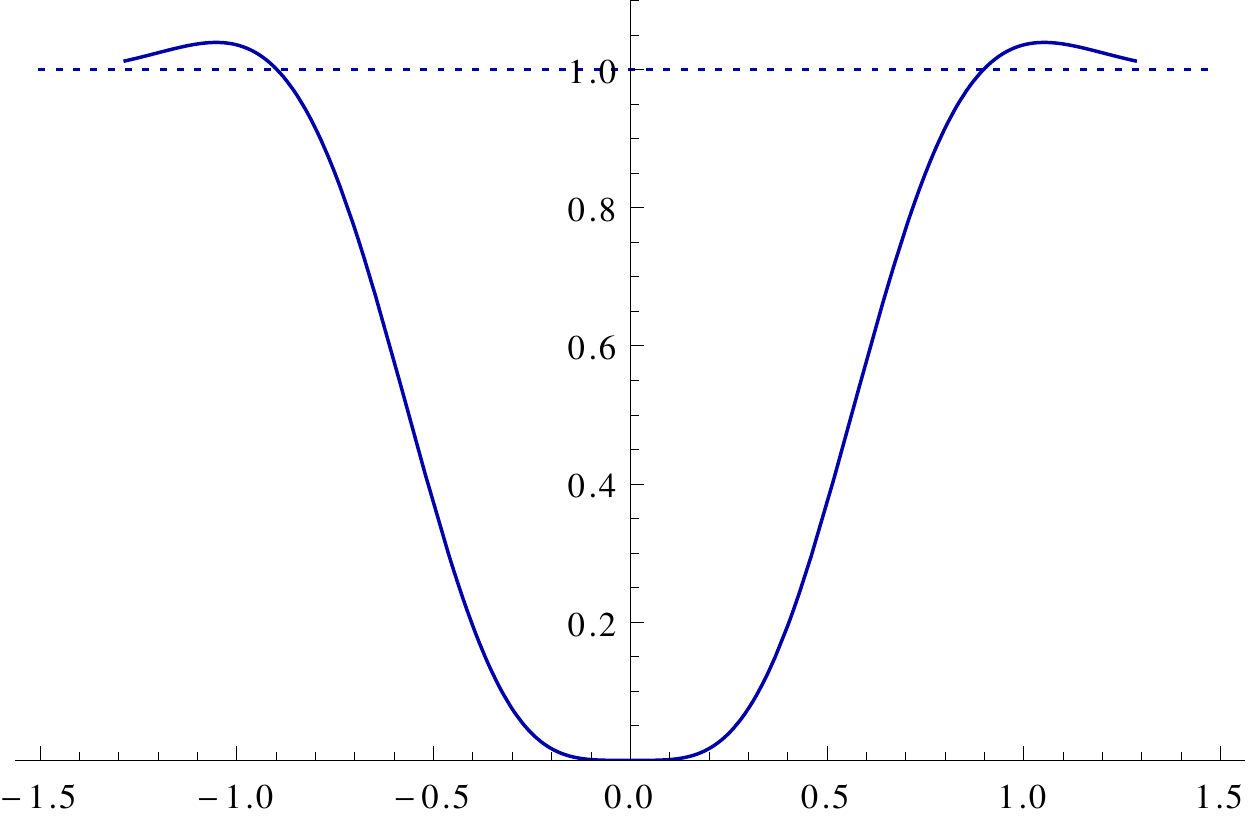}
\caption{The autocorrelation (left) and diffraction (right) density
  (viewed along a line through the origin) of the point process that derives
  from the zeros of the Gaussian random analytic function
  \eqref{eq:GAF}.}\label{fig:GAF}
\end{figure}

We can see from Figure~\ref{fig:GAF} that the diffraction density
exceeds $1$ for $s \approx 1$.  Consequently, as already observed in
\cite{HKPV:2009}, the zero set of the Gaussian random analytic function
$f$ is \emph{not} a determinantal point process.

\begin{proof}[{Proof of Theorem $\ref{thm:GAF}$}]
Let
\begin{equation}\label{eq:function-phi}
   \varphi(u) \, := \, \frac{e^{-u} \big( {-}2+4u-u^2 \big) + 
   e^{-2u} \big( 4-4u-u^2 \big) - 2 e^{-3u}}{(1-e^{-u})^3} 
   \, , \quad u > 0 \ts .
\end{equation}
A straightforward Taylor expansion shows that $\varphi(u) = 1 + o(u)$ 
as $u \to 0$.  Thus, $\varphi(u)$ has a~continuous extension at zero,
with $\varphi(0) = 1$.  Moreover, there exists a constant $C > 0$ such that
$|\varphi(u)| \leq C e^{-u/2}$ for all $u \geq 0$.  It therefore follows
that the function $g(r) = \varphi(\pi r^2)$ (regarded as a radially
symmetric function on $\RR^2$) is~integrable on $\RR^2$.

We can now compute the Fourier transform of $g$.  
In general, when $g \in L^1(\mathbb R^2)$ is radially symmetric, 
which means that it depends only on
the Euclidean norm $r = |x|$, the analogous property holds for
$\widehat{g}$.  
Writing $g(r)$ and $\widehat{g}(s)$ instead of $g(x)$ and $\widehat{g}(t)$, respectively,
and using polar coordinates,
one obtains
\begin{align}
    \widehat{g}(s) \,
  = \int_0^\infty \int_{0}^{2\pi} e^{-2\pi\i rs\cos(\vartheta)} 
      \d\vartheta \, g(r) \, r \d{r} 
  = \, 2\pi \int_0^\infty r g(r) \, J_0(2\pi rs) \d{r} \ts ,
  \label{eq:fourier-rsf}
\end{align}
which is essentially the Hankel transform in one dimension.  Here, we
have employed the classic identity
\[
  \frac{1}{2\pi} \int_{0}^{2\pi} 
   e^{\i z \cos(\vartheta)} \d{\vartheta} \, = \, J_0(z) \ts ,
\]
where $J_0$ is the Bessel function of the first kind of order $0$,
with series expansion
\[
  J_0 (z)\, = \sum_{m=0}^{\infty} (-1)^m \frac{z^{2m}}{4^m (m!)^2} \ts .
\]
Clearly, $J_0(0) = 1$, while $J_0(r) = \myo(r^{-1/2})$ as $r \to
\infty$; compare \cite{abramowitz-stegun:1964}.
Applying the identity \eqref{eq:fourier-rsf} 
to the function $g(r) = \varphi(\pi r^2)$ 
from \eqref{eq:function-g}, we obtain, after a change of variables,
\begin{equation}\label{eq:appendix}
  \widehat{g}(s) \, = \int_0^\infty \varphi(u) \, J_0(\sqrt{4\pi u} s) \d{u}
  \, = \sum_{m=0}^{\infty} \frac{(-1)^m \pi^m s^{2m}}{(m!)^2} 
  \int_0^\infty u^m \varphi(u) \d{u} \ts .
\end{equation}
The exchange of integration and summation 
is justified by dominated convergence, 
using the estimate $|\varphi(u)| \leq C e^{-u/2}$.

\begin{lemma}\label{lemma:appendix}
  The function\/ $I : [0,\infty) \to \RR$ defined by\/ $I(\alpha) :=
  \int_0^\infty u^\alpha \varphi(u) \d{u}$ is continuous on\/ $[0,\infty)$,
  with
\[
  I(\alpha) \, = \, \begin{cases}
       1 , & \text{if\/} \ \alpha = 0, \\
       \alpha (1-\alpha) \Gamma(\alpha+1) 
       \zeta(\alpha+1) , & \text{if\/} \ \alpha > 0, 
\end{cases}
\]
where\/ $\Gamma$ is the gamma function and\/ $\zeta$ is Riemann's zeta
function.
\end{lemma}

We thus have $I(0) = 1$, $I(1) = 0$ and $I(m) = -m(m-1) \, m! \,
\zeta(m+1)$ for $m \in \NN$, $m \geq 2$.  Inserting this into
\eqref{eq:appendix} gives
\[
  \widehat{g}(s) \, = \, 1 + 
  \sum_{m=2}^{\infty} \frac{(-1)^{m+1} }{(m-2)!} \,
  (\pi s^2)^{m} \ts \zeta(m+1) \ts ,
\]
as claimed.
\end{proof}

\begin{proof}[Proof of Lemma $\ref{lemma:appendix}$]
  Using the estimate $|\varphi(u)| \leq C e^{-u/2}$, it is straightforward
  to see that $u \mapsto u^\alpha \varphi(u)$ is integrable for any $\alpha
  \geq 0$ and, by dominated convergence, $\alpha \mapsto I(\alpha)$ is
  continuous at any $\alpha \geq 0$.

  Observe that $(1-e^{-u})^{-3} = \sum_{n=0}^{\infty} \tfrac12
  (n+1)(n+2)\ts e^{-nu}$ for any $u > 0$.  Inserting this into
  \eqref{eq:function-phi} gives
\[
   \varphi(u) \, = \,  - \sum_{n=1}^{\infty} (n^2 u^2 - 4nu + 2)\ts e^{-nu}
\]
for any $u > 0$.  Since $\int_0^\infty u^{x-1} e^{-u} \d{u} =
\Gamma(x)$ and $\Gamma(x+1) = x \Gamma(x)$ for any $x > 0$, we~find
for any fixed $\alpha > 0$ that
\[
  I(\alpha) \, = \, - \sum_{n=1}^{\infty} 
  \frac{\Gamma(\alpha+3)}{n^{\alpha+1}}
  + 4 \sum_{n=1}^{\infty} \frac{\Gamma(\alpha+2)}{n^{\alpha+1}}
  - 2 \sum_{n=1}^{\infty} \frac{\Gamma(\alpha+1)}{n^{\alpha+1}}
  \, = \, \alpha (1-\alpha) \Gamma(\alpha+1) \zeta(\alpha+1) \ts .
\]
Here, the termwise integration is justified by dominated convergence
as long as $\alpha > 0$.

Recalling that $\alpha \zeta(\alpha + 1) = 1 + o(1)$ as $\alpha \to 0$
and using the continuity of $\alpha \mapsto I(\alpha)$ at $\alpha = 0$,
we finally obtain
\[
  I(0) \, = \, \lim_{\alpha \to 0} I(\alpha) 
  \, = \, \Gamma(1) \, = \, 1 \ts ,
\]
which completes the argument.
\end{proof}

Although we have seen in Proposition \ref{prop:GAF-mixing} that 
the point process of zeros of the Gaussian random analytic function $f$ 
is mixing, any realisation of it displays an amazing amount of structure. 
Indeed, the zero set can be given a~remarkable visual interpretation 
as tilings. The details can be found in \cite[Chapter 8]{HKPV:2009},
but briefly, the~tilings arise as the basins of descent
of the `potential function' $u$ on~$\CC$
defined by
\[
    u(z) := \log |f(z)| - \tfrac{1}{2} |z|^2 \,.
\]
This function goes to $-\infty$ at the zeros of $f$, and if one follows
the gradient curves defined by the equation
\[
    \frac{dZ(t)}{dt} = \nabla u(Z(t)) \,,
\]
they can be thought of as paths of descent under $u$ 
as a `gravitational' attraction. The~basins of attraction
of each zero of $f$ then lead to a tiling. Remarkably,
these tiles almost surely have the same area $\pi/L$
and are bounded by finitely many smooth curves; 
see \cite{sodin-tsirelson:2006} or \cite[Theorem~8.2.7]{HKPV:2009}.
Technically, this is an example of an \emph{allocation},
by which area is associated to each point of the point process.

\begin{figure}
\includegraphics[width=0.6\textwidth]{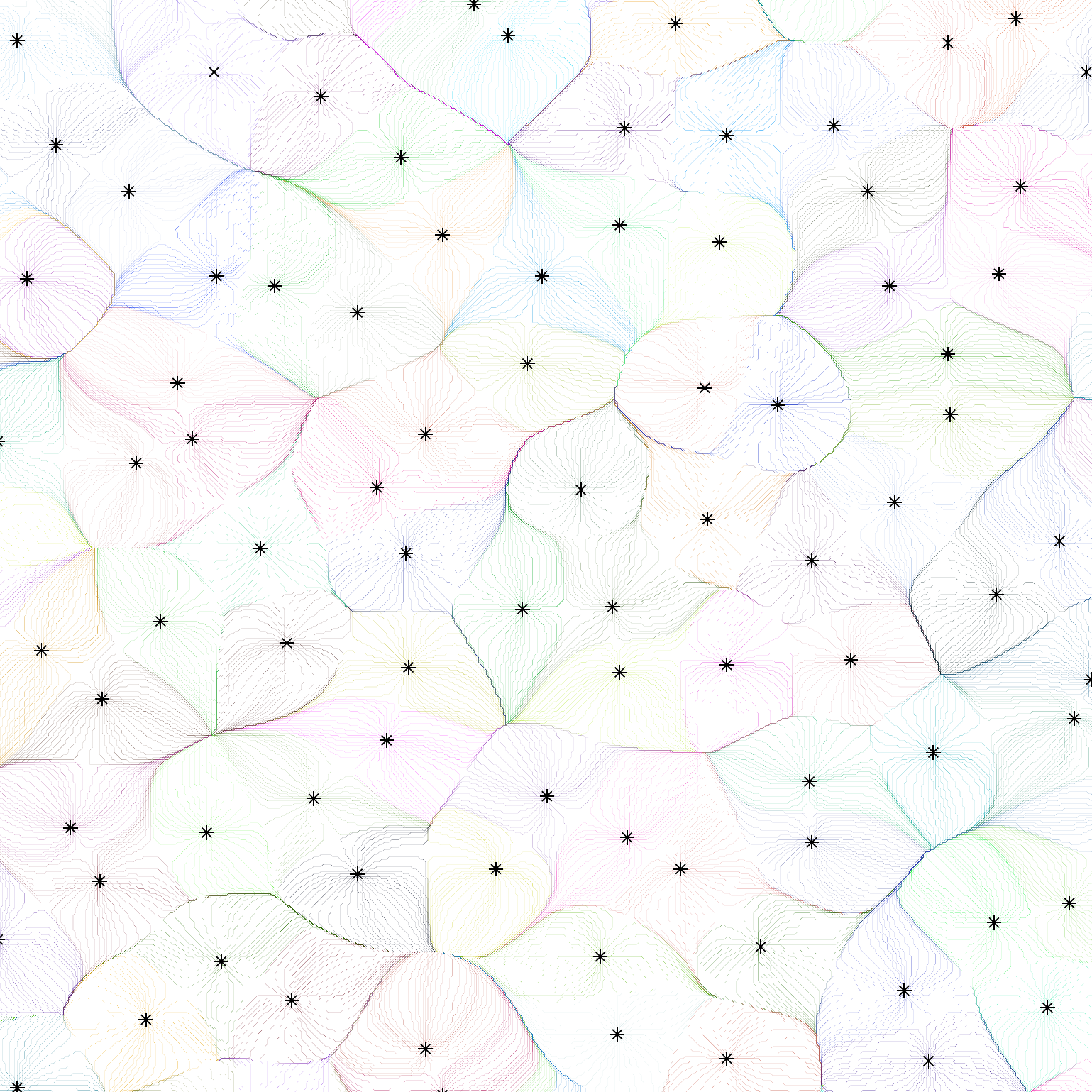}

\caption{The point process of zeros of the Gaussian random analytic function
\eqref{eq:GAF} and the associated basins of attraction.
Thanks to Manjunath \mbox{Krishnapur} and Ron \mbox{Peled}
for making the MatLab program available to us.}
\label{fig:basins}
\end{figure}

Clearly, we have no long-range order that manifests itself
as non-trivial Bragg peaks, and also none that would lead to
singular continuous components. Yet, there are lots of 
remarkable patterns that almost repeat (at random locations,
of course), and one might wonder to what extent spectral theory
is able to capture such features. This might be an interesting point
for future investigations.

\section*{Acknowledgements}
We thank Manjunath Krishnapur for drawing our attention
to Gaussian analytic functions.
We also thank the reviewers for their useful comments.
This project was supported by the German Research Foundation (DFG),
within the CRC~701. Robert V.\@ Moody was also supported
by the Natural Sciences and Engineering Research Council
of Canada.


\begin{thebibliography}{99}

\bibitem{abramowitz-stegun:1964}
Abramowitz, M.; Stegun, I.~A.:
\emph{Handbook of Mathematical Functions.}
Dover Publications (1972), New~York.

\bibitem{AGZ:2010}
Anderson, G.; Guionnet, A.; Zeitouni, O.:
\emph{An Introduction to Random Matrices.}
Cambridge University Press (2010), Cambridge.

\bibitem{baake-birkner-moody:2010}
Baake, M.; Birkner, M.; Moody, R.~V.:
Diffraction of stochastic point sets: explicitly computable examples.
\textit{Commun. Math. Phys.} \textbf{293} (2010), 611--660.
arXiv:0803.1266

\bibitem{baake-grimm:2013}
Baake, M.; Grimm, U.:
\emph{Aperiodic Order. Volume 1: A Mathematical Invitation.}
Cambridge University Press (2013), Cambridge.

\bibitem{baake-hoeffe:2000}
Baake, M.; H\"offe, M.:
Diffraction of random tilings: some rigorous results.
\textit{J. Stat. Phys.} \textbf{99} (2000), 219--261.
arXiv:math-ph/9901008

\bibitem{baake-koesters:2011}
Baake, M.; K\"osters, H.:
Random point sets and their diffraction.
\textit{Philos. Mag.} \textbf{91} (2011), 2671--2679.
arXiv:1007.3084

\bibitem{baake-lenz-vanenter:2013}
Baake, M.; Lenz, D.; van Enter, A.:
Dynamical versus diffraction spectrum for structures
    with finite local complexity.
arXiv:1307.7518

\bibitem{berg-forst:1975}
Berg, C., Forst, G.: 
\emph{Potential Theory on Locally Compact Abelian Groups.}
Springer (1975), Berlin.

\bibitem{CFS}
Cornfeld, I.~P.; Fomin, S.~V.; Sinai, Ya.~G.:
\emph{Ergodic Theory.}
Springer (1982), New~York.

\bibitem{DV1}
Daley, D.~J.; Vere-Jones, D.:
\emph{An Introduction to the Theory of Point Processes.
    Volume I: Elementary Theory and Methods,} 2nd edition.
Springer (2003), New~York.

\bibitem{DV2}
Daley, D.~J.; Vere-Jones, D.:
\emph{An Introduction to the Theory of Point Processes.
    Volume I\hspace{-1pt}I: General Theory and Structure,} 2nd edition.
Springer (2008), New~York.

\bibitem{gouere:2003}
Gouer\'e, J.-B.: 
Diffraction and Palm measure of point processes. 
\textit{Comptes Rendus Acad. Sci.} \textbf{342} (2003), 141--146.
arXiv:math/0208064

\bibitem{HKPV:2009}
Hough, J.~B.; Krishnapur, M.; Peres, Y.; Virag, B.:
\textit{Zeros of Gaussian Analytic Functions and Determinantal Point Processes.}
American Mathematical Society (2009), Providence, RI.

\bibitem{lavancier-moller-rubak:2012}
Lavancier, F.; M{\o}ller, J; Rubak, E.:
Statistical aspects of determinantal point processes.
arXiv:1205.4818

\bibitem{lenz-strungaru:2009}
Lenz, D.; Strungaru, N.:
Pure point spectrum for measure dynamical systems on locally compact abelian groups.
\emph{J. Math. Pures Appl.} \textbf{92} (2009), 323--341.
arXiv:0704.2489

\bibitem{macchi:1975}
Macchi, O.:
The coincidence approach to stochastic point processes.
\textit{Adv. Appl. Probab.} \textbf{7} (1975), 83--122.

\bibitem{queffelec}
Queff\'elec, M.:
\emph{Substitution Dynamical Systems -- Spectral Analysis}, 2nd edition.
Springer (2010), Berlin.

\bibitem{sodin-tsirelson:2006}
Sodin, M.; Tsirelson, B.
Random complex zeros II. Perturbed lattice.
\textit{Israel J. Math.} \textbf{152} (2006), 105--124.
arXiv:math/0309449

\bibitem{soshnikov:2000}
Soshnikov, A.:
Determinantal random point fields.
\textit{Russian Math. Surveys} \textbf{55} (2000), 923--975.
arXiv:math/0002099

\end{thebibliography}
\end{document}